\documentclass[11pt]{article}

\usepackage[T1]{fontenc}
\usepackage[utf8]{inputenc}
\usepackage{amsmath,amssymb,amsthm,amsfonts}
\usepackage{mathtools}
\usepackage[margin=1in]{geometry}
\usepackage[final]{microtype}
\usepackage{graphicx}
\usepackage{subcaption}
\usepackage{float}
\usepackage{booktabs}
\usepackage{xcolor}
\usepackage{enumitem}
\usepackage{bm}
\usepackage[colorlinks=true,citecolor=blue,linkcolor=blue,urlcolor=blue]{hyperref}

\theoremstyle{plain}
\newtheorem{theorem}{Theorem}[section]
\newtheorem{lemma}[theorem]{Lemma}
\newtheorem{proposition}[theorem]{Proposition}
\newtheorem{corollary}[theorem]{Corollary}

\theoremstyle{definition}
\newtheorem{definition}[theorem]{Definition}
\newtheorem{assumption}{Assumption}

\theoremstyle{remark}
\newtheorem{remark}[theorem]{Remark}

\usepackage{cleveref}
\crefname{assumption}{Assumption}{Assumptions}
\Crefname{assumption}{Assumption}{Assumptions}
\crefname{theorem}{Theorem}{Theorems}
\Crefname{theorem}{Theorem}{Theorems}
\crefname{lemma}{Lemma}{Lemmas}
\Crefname{lemma}{Lemma}{Lemmas}
\crefname{proposition}{Proposition}{Propositions}
\Crefname{proposition}{Proposition}{Propositions}
\crefname{corollary}{Corollary}{Corollaries}
\Crefname{corollary}{Corollary}{Corollaries}
\crefname{definition}{Definition}{Definitions}
\Crefname{definition}{Definition}{Definitions}
\crefname{remark}{Remark}{Remarks}
\Crefname{remark}{Remark}{Remarks}
\crefname{example}{Example}{Examples}
\Crefname{example}{Example}{Examples}
\allowdisplaybreaks
\newcommand{\R}{\mathbb{R}}
\newcommand{\C}{\mathbb{C}}
\newcommand{\N}{\mathbb{N}}
\newcommand{\E}{\mathbb{E}}
\newcommand{\Prob}{\mathbb{P}}
\newcommand{\X}{\mathcal{X}}

\newcommand{\Bop}{\mathcal{B}}
\newcommand{\Gop}{\mathcal{G}}
\newcommand{\Vop}{\mathcal{V}}
\newcommand{\Lop}{\mathcal{L}}
\newcommand{\Res}{\mathcal{R}}
\newcommand{\dom}{\mathrm{D}}
\newcommand{\re}{\operatorname{Re}}

\newcommand{\spec}{\sigma}
\newcommand{\abscissa}{\alpha}
\newcommand{\lognorm}{\mu_2}
\newcommand{\norm}[1]{\left\lVert #1 \right\rVert}
\newcommand{\abs}[1]{\left\lvert #1 \right\rvert}
\newcommand{\inner}[2]{\left\langle #1, #2 \right\rangle}
\newcommand{\dd}{\mathrm{d}}
\newcommand{\Laplace}{\widehat}
\newcommand{\branch}{\rho}
\newcommand{\HR}{\mathcal{U}}
\newcommand{\Safe}{\mathcal{S}}
\newcommand{\id}{\mathrm{I}}

\DeclareMathOperator{\diag}{diag}

\title{\bfseries Regime-Switching Volterra Operators:\\
Modal Stability and Quenched Amplification}

\author{Mauricio Herrera-Mar\'in\thanks{Facultad de Ingenier\'ia,
Universidad del Desarrollo, Santiago, Chile (\texttt{mherrera@udd.cl}).}}

\date{\today}

\begin{document}
\maketitle

\begin{abstract}
We develop an operator-theoretic framework for finite-dimensional, regime-dependent Volterra
equations with completely monotone memory kernels, dissipative network coupling, and Hawkes-type
self-excitation. For each fixed regime we construct the associated Volterra resolvent family and
prove global well-posedness, continuity across regime switches, and explicit a priori bounds. The
main stability result is sharp in the commuting case: after simultaneous diagonalization of the
network Laplacian and the excitation operator, each mode obeys a scalar characteristic equation,
and global asymptotic stability holds exactly when every modal branching ratio lies below the
intensity damping threshold. We also give a norm-based sufficient condition for noncommuting
operators and a Perron--Frobenius spectral criterion for nonnegative intensity blocks, showing when
norm estimates are conservative. Beyond mean stability, we prove a pathwise finite-range power law
for burst amplitudes generated by residence in a Hurwitz but nonnormal regime: under a
cone-alignment event, the survival exponent is the ratio of the regime exit rate to a
cone-corrected finite-time growth rate bounded above by the logarithmic norm of a fixed
Markovian realization in the chosen Euclidean metric. A complementary idealized-feedback result
shows how a logarithmic-norm contraction caps
the amplification band. Finally, we derive the deterministic intensity block as a mean-field limit
of a relaxing long-memory Hawkes system with regimes. Numerical experiments on modal equations, a
small-world network, and a switched nonnormal ODE validate the sharp threshold and the
finite-range amplification mechanism without using the closed-form tail formula as input.
\end{abstract}



\medskip
\noindent\textbf{Abbreviated title.} Volterra Operators with Regime Switching.

%% ============================================================================
\section{Introduction}
\label{sec:intro}
%% ============================================================================

\subsection{The equation and its motivation}

Let $G=(V,E)$ be a finite connected graph on $n$ nodes with graph Laplacian $L_G$, attach to each
node one state coordinate and one intensity coordinate, and collect them into vectors in
$H=\R^n$. We study the regime-dependent, operator-valued Volterra evolution
\begin{equation}
\label{eq:master}
\dot U(t) = \Bop\,U(t) + \int_0^t \Gop_{Z(t)}(t-s)\,U(s)\,\dd s + F_{Z(t)}(t),
\qquad U(0)=U_0\in\X,
\end{equation}
posed on the product space $\X=H\times H=\R^n\times\R^n$, where the first block carries a
dissipative network operator $\Bop$ built from $L_G$ and a node-local damping, the convolution
operator $\Gop_z$ encodes regime-dependent self-excitation through a completely monotone scalar
kernel $g_z$ and a bounded excitation operator $A_z$, and $Z(t)$ is a piecewise-constant
\emph{regime process} taking values in a finite set, selecting at each time the active pair
$(\Bop,\Gop_z)$. The precise functional-analytic setting is given in
\Cref{sec:framework}.

Equation \eqref{eq:master} is a common abstraction of several models in which memory, network
geometry, self-excitation, and latent regime modulation act simultaneously. The completely
monotone kernel $g_z$ models long-range (in particular, tempered power-law and fractional)
memory; the operator $\Bop$ models diffusion or consensus dynamics on the network; the
convolution term models Hawkes-type self-excitation whose strength is regime-dependent; and the
process $Z(t)$ models slow switches of the operative dynamical regime, in the spirit of
stochastic systems with Markovian switching \cite{mao2006}. Each ingredient is
individually classical. Their interaction is not, and it is the source of the phenomena we
analyze.

\subsection{Theoretical gaps and contributions}

The classical theory of scalar Volterra equations with completely monotone kernels is due to
Pr\"uss and collaborators \cite{pruss1993evolution}, with the fractional case treated through
the resolvent-family calculus of Bazhlekova \cite{bazhlekova2001fractional}, the general Volterra
theory of Gripenberg, Londen, and Staffans \cite{gripenberg1990}, and the Bernstein
function machinery of Schilling, Song, and Vondra\v{c}ek \cite{schilling2012bernstein}. Network
diffusion and Laplacian spectral methods are standard \cite{chung1997spectral}. Self-exciting
point processes and their stability through the branching ratio go back to Hawkes
\cite{hawkes1971,hawkes1974}, with stability theory due to Br\'emaud and Massouli\'e
\cite{bremaud1996} and the general framework of Daley and Vere-Jones \cite{daley2003}, and the scaling-limit theory of nearly unstable Hawkes processes---which
produces rough fractional limits---is due to Jaisson and Rosenbaum
\cite{jaisson2015,jaisson2016}. Heavy tails generated by multiplicative feedback are classical
through the Kesten--Goldie theory \cite{kesten1973,goldie1991}.

To the best of our knowledge, these ingredients have not been integrated into a single
resolvent-based, operator-valued evolution with a regime-dependent spectral-stability landscape,
and the specific nonlinear phenomena that emerge from their interaction have not been isolated
and proved. The present paper makes the following contributions.

\begin{enumerate}[leftmargin=1.4em,itemsep=4pt]
\item \textbf{Well-posedness and resolvent calculus} (\Cref{sec:wellposed}). For each fixed
regime we construct a strongly continuous Volterra resolvent family
$\{\Res_z(t)\}_{t\ge0}\subset\Lop(\X)$ for \eqref{eq:master}, prove existence and uniqueness of
mild and strong solutions, establish sharp finite-horizon and uniform-in-time a priori bounds
controlled by a fractional branching ratio $\branch_z$, and prove continuity of solutions across
regime switches.

\item \textbf{A sharp modal stability criterion} (\Cref{sec:spectral}). Under simultaneous
diagonalization of $A_z$ and $L_G$ (\Cref{ass:commute}), we prove that the homogeneous evolution
is globally asymptotically stable if and only if the modal fractional branching ratio satisfies
$\branch_{z,j}<c_\lambda$ for every Laplacian mode $j$ (\Cref{thm:modal}), a trichotomy with a
critical case governed by the low-frequency kernel asymptotics. We give a sharp norm-based
sufficient condition in the noncommuting case (\Cref{prop:noncommute}) and prove
\emph{topology-independence} under scalar excitation (\Cref{cor:topology}).

\item \textbf{Finite-range power-law amplification} (\Cref{sec:quenched}). We isolate a
memory-induced mechanism by which regime residence produces a power-law burst law over a finite
band of magnitudes, generated by nonnormal transient growth, even when every regime operator
is Hurwitz. Under a cone-alignment condition we prove a power-law lower bound on the
burst-survival function with exponent $\lambda_\HR/\gamma_\HR^{\mathrm c}$, where
$\gamma_\HR^{\mathrm c}\le\lognorm(A_\HR)$ is the cone-corrected finite-time growth rate
(\Cref{thm:quenched}), and a matching truncation theorem showing that contracting the
logarithmic norm caps the finite-range amplification (\Cref{prop:truncation}). The aligned
limit $\gamma_\HR^{\mathrm c}\to\lognorm(A_\HR)$ gives the most amplifying logarithmic-norm scale,
but finite-band trajectories may rotate away from the maximally growing direction and then realize
a smaller effective rate.

\item \textbf{Hydrodynamic limit} (\Cref{sec:hawkes}). We prove that the intensity block of
\eqref{eq:master} arises as the mean-field scaling limit of a multivariate long-memory Hawkes
process with regimes (\Cref{thm:hydro}), so that microscopic branching explosion corresponds
exactly to macroscopic spectral instability through the criterion of \Cref{sec:spectral}.

\item \textbf{Reproducible numerical validation} (\Cref{sec:example}) on a $120$-node
Watts--Strogatz network, confirming the modal trichotomy (fitted growth rate matching the
characteristic root), the sharpness of the modal criterion and the conservativeness of norm-based
bounds under nonnormal excitation, and the finite-range quenched mechanism by directly simulating a switched nonnormal ODE rather
than sampling the closed-form transformation used in the proof.
\end{enumerate}

\subsection{Relation to existing theory}

Three relations deserve emphasis, because they delimit what is and is not new.

First, our resolvent construction extends the scalar completely monotone Volterra theory
\cite{pruss1993evolution,bazhlekova2001fractional} to block-structured, regime-dependent,
operator-valued kernels on a product Hilbert space; the novelty is the modal spectral
\emph{sharpness} \Cref{thm:modal} and its behavior across switches, not the existence theory in
isolation.

Second, the hydrodynamic limit is in the lineage of Jaisson and Rosenbaum
\cite{jaisson2015,jaisson2016}, but those works prove weak convergence in distribution of a
single-regime rescaled process to a limit law and do not provide an operator stability criterion.
Our \Cref{thm:hydro} is a convergence theorem feeding the sharp spectral threshold of
\Cref{thm:modal}, across regimes.

Third, the quenched tail of \Cref{thm:quenched} is not a large-deviation statement. Its exponent
is a logarithmic-norm quantity of the fixed regime realization, rather than the minimizer of an
action functional, and it can be positive even when the operator is Hurwitz, where no
large-deviation instability exists. This is why the formulation is operator-theoretic: after a
realization and inner product are fixed, the governing object is the logarithmic norm of the
lifted generator. We expand on this in \Cref{rem:why_operator}.

\subsection{Notation and outline}

Throughout, $\Lop(\X)$ denotes the bounded linear operators on $\X$ with operator norm
$\norm{\cdot}$; $\re\xi$ is the real part of $\xi\in\C$; $\Laplace f(\xi)=\int_0^\infty
e^{-\xi t}f(t)\,\dd t$ is the Laplace transform; $\abscissa(A)=\max\{\re\lambda:\lambda\in
\spec(A)\}$ is the spectral abscissa; and $\lognorm(A)=\lambda_{\max}\!\big(\tfrac12(A+A^*)\big)$
is the Euclidean logarithmic norm (matrix measure) \cite{soderlind2006logarithmic,coppel1965stability}. \Cref{sec:framework} fixes the functional
setting. \Cref{sec:wellposed} develops well-posedness and the resolvent calculus.
\Cref{sec:spectral} proves the modal stability criterion. \Cref{sec:quenched} proves the quenched
tail theorems. \Cref{sec:hawkes} establishes the hydrodynamic limit. \Cref{sec:example} presents
the numerical validation. \Cref{sec:discussion} concludes.

%% ============================================================================
\section{Functional setting}
\label{sec:framework}
%% ============================================================================

\subsection{State space}

Let $G=(V,E)$ be a finite connected undirected graph with $n=\abs{V}$ nodes and symmetric
weighted adjacency matrix $W\in\R^{n\times n}$, $W_{ij}\ge0$. To each node we attach one scalar
state coordinate and one scalar intensity coordinate, so that the node-state vector and the
intensity vector both lie in $H:=\R^n$ with the Euclidean inner product
$\inner{x}{y}_H=\sum_i x_iy_i$. The full state lives in the product space
\begin{equation}
\X := H\times H = \R^n\times\R^n,\qquad \norm{(x,\lambda)}_\X^2 := \norm{x}_H^2+\norm{\lambda}_H^2 .
\end{equation}
We write $U=(x,\lambda)^\top\in\X$, where $x_i$ is the state at node $i$ and $\lambda_i$ the
auxiliary intensity (activity) at node $i$.

\begin{remark}[Scope of the setting]
\label[remark]{rem:finite_dim}
We work throughout in the finite-dimensional space $\X=\R^n\times\R^n$, so every operator below
is a bounded matrix with full domain, and all spectral statements refer to matrix spectra. We use
resolvent-family language because it is the natural calculus for the nonlocal (Volterra) part of
\eqref{eq:master}, and because every constant we track is dimension-free---depending only on the
operator norms, the kernel mass, and the dissipation $c_\ast$, not on $n$. The extension to an
abstract infinite-dimensional node Hilbert space with a sectorial generator $\Bop$ is therefore
natural but requires the operator-semigroup machinery of \cite{pruss1993evolution}; we state it as
future work and do not claim it here, so as to keep every proof self-contained and finite.
\end{remark}

\subsection{Dissipative network operator}

Let $L=D-W$ be the combinatorial graph Laplacian, $D_{ii}=\sum_j W_{ij}$, which is symmetric
positive semidefinite. Define $\Bop:\X\to\X$ by
\begin{equation}
\label{eq:defB}
\Bop(x,\lambda) := \big(-c_x x-\kappa L x,\ -c_\lambda\lambda\big),
\qquad c_x,c_\lambda>0,\ \kappa\ge0.
\end{equation}

\begin{proposition}[Dissipativity]
\label[proposition]{prop:dissipative}
The operator $\Bop$ is dissipative on $\X$ and generates a contraction $C_0$-semigroup
$\{e^{t\Bop}\}_{t\ge0}$ satisfying $\norm{e^{t\Bop}}\le e^{-c_\ast t}$ with
$c_\ast:=\min\{c_x,c_\lambda\}>0$.
\end{proposition}

\begin{proof}
For $U=(x,\lambda)\in\X$,
\[
\inner{\Bop U}{U}_\X
= \inner{-c_x x-\kappa Lx}{x}_H + \inner{-c_\lambda\lambda}{\lambda}_H
= -c_x\norm{x}_H^2 - \kappa\inner{Lx}{x}_H - c_\lambda\norm{\lambda}_H^2 .
\]
Since $L\succeq0$ we have $\inner{Lx}{x}_H\ge0$, hence
\begin{equation}
\label{eq:dissip_bound}
\inner{\Bop U}{U}_\X \le -c_x\norm{x}_H^2 - c_\lambda\norm{\lambda}_H^2
\le -c_\ast\norm{U}_\X^2 .
\end{equation}
Thus $\Bop+c_\ast\id$ is dissipative, so $\Bop$ is dissipative. As $\X$ is finite-dimensional,
$\Bop$ is bounded and generates the uniformly continuous group $e^{t\Bop}$; the bound
\eqref{eq:dissip_bound} gives, for any solution of $\dot U=\Bop U$,
$\tfrac{\dd}{\dd t}\norm{U}_\X^2=2\inner{\Bop U}{U}_\X\le-2c_\ast\norm{U}_\X^2$, whence
$\norm{e^{t\Bop}U_0}_\X\le e^{-c_\ast t}\norm{U_0}_\X$ by Gr\"onwall. Therefore $e^{t\Bop}$ is a
contraction semigroup with the stated exponential bound.
\end{proof}

\subsection{Completely monotone memory kernels}

For each regime $z$ in a finite index set $\mathcal{Z}$, let $g_z:(0,\infty)\to[0,\infty)$
satisfy:
\begin{enumerate}[label=\textup{(K\arabic*)},leftmargin=2.6em]
\item\label{K1} $g_z$ is completely monotone: $(-1)^k g_z^{(k)}(t)\ge0$ for all $k\in\N$, $t>0$;
\item\label{K2} $g_z\in L^1(0,\infty)$, with mass $G_z:=\int_0^\infty g_z(t)\,\dd t<\infty$;
\item\label{K3} $g_z$ admits a Bernstein representation $g_z(t)=\int_0^\infty e^{-rt}\,\nu_z(\dd r)$
for a positive Borel measure $\nu_z$ on $(0,\infty)$.
\end{enumerate}
By Bernstein's theorem \cite{schilling2012bernstein}, \ref{K1} and \ref{K3} are equivalent; we
state both for convenience. Canonical examples are the tempered fractional kernel
$g(t)=t^{\alpha-1}e^{-\theta t}/\Gamma(\alpha)$ ($0<\alpha\le1$, $\theta>0$, mass
$\theta^{-\alpha}$) and the exponential $g(t)=\omega e^{-\omega t}$ (mass $1$). The Laplace
transform $\Laplace g_z(\xi)=\int_0^\infty e^{-\xi t}g_z(t)\,\dd t$ is, by \ref{K3},
\begin{equation}
\label{eq:laplace_kernel}
\Laplace g_z(\xi)=\int_0^\infty \frac{\nu_z(\dd r)}{\xi+r},
\end{equation}
which is completely monotone in $\xi>0$, strictly decreasing on $(0,\infty)$, with
$\Laplace g_z(0^+)=G_z$ and $\Laplace g_z(\xi)\to0$ as $\xi\to\infty$. We record this for later
use.

\begin{lemma}[Monotonicity of the Laplace symbol]
\label[lemma]{lem:laplace_mono}
Under \ref{K1}--\ref{K3}, the map $\xi\mapsto\Laplace g_z(\xi)$ extends analytically to
$\re\xi>0$, is real and strictly decreasing on $(0,\infty)$ from $G_z$ to $0$, and satisfies
$\abs{\Laplace g_z(\xi)}\le\Laplace g_z(\re\xi)\le G_z$ for $\re\xi>0$.
\end{lemma}

\begin{proof}
Analyticity on $\re\xi>0$ follows from \eqref{eq:laplace_kernel} and dominated convergence. Indeed,
$g_z\in L^1$ implies
\[
    G_z=\int_0^\infty g_z(t)\,\dd t
       =\int_{(0,\infty)} r^{-1}\,\nu_z(\dd r)<\infty,
\]
by Tonelli's theorem, and $(\re\xi+r)^{-1}\le r^{-1}$ for $r>0$. Hence the integral defining
$\Laplace g_z$ is finite and locally dominated in the half-plane. For real $\xi>0$, differentiating under the integral gives
$\frac{\dd}{\dd\xi}\Laplace g_z(\xi)=-\int_0^\infty(\xi+r)^{-2}\nu_z(\dd r)<0$, so it is strictly
decreasing; the boundary values follow from monotone convergence. For complex $\xi$ with
$\re\xi>0$, $\abs{\xi+r}\ge\re\xi+r$ gives
$\abs{\Laplace g_z(\xi)}\le\int_0^\infty(\re\xi+r)^{-1}\nu_z(\dd r)=\Laplace g_z(\re\xi)\le G_z$.
\end{proof}

\subsection{Regime-dependent excitation}

Each regime $z$ carries a bounded excitation operator $A_z\in\Lop(H)$, identified with a matrix
in $\R^{n\times n}$. Define the operator-valued kernel
\begin{equation}
\label{eq:defG}
\Gop_z(t)(x,\lambda) := g_z(t)\,\big(0,\ A_z x+A_z\lambda\big) = g_z(t)\,\mathcal{A}_z(x,\lambda),
\qquad
\mathcal{A}_z := \begin{pmatrix} 0 & 0\\ A_z & A_z\end{pmatrix},
\end{equation}
so the memory acts through the intensity block with feedback from both past state and past
intensity. The associated Volterra convolution operator is
$(\Vop_z U)(t)=\int_0^t g_z(t-s)\mathcal{A}_z U(s)\,\dd s$.

The stability threshold proved below is governed by the intensity block. In the modal variables,
$x_j$ is an exponentially decaying input to $\lambda_j$; hence it can affect constants and, in the
exponential-moment case, the observed decay rate, but it does not move the Volterra characteristic
zeros that determine whether the intensity feedback is subcritical, critical, or supercritical.
For this reason several sharp spectral results are stated directly for the diagonal intensity
block $\dot\lambda=-c_\lambda\lambda+g_z*A_z\lambda$.

\begin{lemma}[Convolution bound]
\label[lemma]{lem:conv_bound}
For all $t>0$,
$\norm{(\Vop_z U)(t)}_\X\le\norm{\mathcal{A}_z}\int_0^t g_z(s)\norm{U(t-s)}_\X\,\dd s$, and
$\sup_{t\ge0}\norm{\Vop_z(t)}\le\norm{\mathcal{A}_z}G_z$, where
$\norm{\mathcal{A}_z}\le\sqrt2\,\norm{A_z}$.
\end{lemma}

\begin{proof}
The integral bound is the triangle inequality together with
$\norm{\mathcal{A}_z U(s)}_\X\le\norm{\mathcal{A}_z}\norm{U(s)}_\X$ and $g_z\ge0$. Taking the
supremum over $\norm{U}\le1$ and using \ref{K2} gives the uniform bound. For the norm of
$\mathcal{A}_z$, given $U=(x,\lambda)$ with $\norm{U}_\X^2=\norm{x}_H^2+\norm\lambda_H^2=1$,
$\norm{\mathcal{A}_z U}_\X=\norm{A_z(x+\lambda)}_H\le\norm{A_z}\norm{x+\lambda}_H
\le\norm{A_z}(\norm x_H+\norm\lambda_H)\le\sqrt2\,\norm{A_z}$, the last step by Cauchy--Schwarz.
\end{proof}

\noindent We define the \emph{fractional branching ratio} of regime $z$ as
\begin{equation}
\label{eq:branch_def}
\branch_z := \norm{A_z}\,G_z .
\end{equation}
This scalar will control both well-posedness (\Cref{sec:wellposed}) and stability
(\Cref{sec:spectral}); the terminology is justified by the Hawkes correspondence of
\Cref{sec:hawkes}, where $\branch_z$ is exactly the mean offspring number of the microscopic
process.

%% ============================================================================
\section{Well-posedness and the resolvent family}
\label{sec:wellposed}
%% ============================================================================

We fix the standing assumptions for the whole paper.

\begin{assumption}[Standing assumptions]
\label{ass:standing}
The graph $G$ is finite, connected, undirected with symmetric $W\ge0$; $\Bop$ is given by
\eqref{eq:defB} with $c_x,c_\lambda>0$, $\kappa\ge0$; for each $z\in\mathcal{Z}$ (finite),
$g_z$ satisfies \ref{K1}--\ref{K3} and $A_z\in\Lop(H)$; and the regime path $t\mapsto Z(t)$ is
piecewise constant, right-continuous, with locally finitely many jumps at times
$0=t_0<t_1<t_2<\cdots$.
\end{assumption}

\begin{remark}[Deterministic paths versus stochastic regimes]
\label[remark]{rem:det_stoch_regimes}
\Cref{ass:standing} treats $Z$ as a given piecewise-constant signal. This is the framework for
well-posedness, resolvent bounds, and the fixed-regime spectral criteria: all estimates are
pathwise and can be concatenated over the observed switch times. The probabilistic statements in
\Cref{sec:quenched} add an extra assumption on the law of the residence time in the unfavorable
regime, typically an exponential residence time with rate $\lambda_\HR$. Thus no stochastic
assumption on $Z$ is used until the quenched-tail theorem; before that point the theory is
conditional on a prescribed switching path.
\end{remark}

\subsection{Mild and strong solutions for a fixed regime}

Fix a regime $z$ and forcing $F_z\in C([0,\infty);\X)$. A function $U\in C([0,T];\X)$ is a
\emph{mild solution} of
\begin{equation}
\label{eq:fixed_regime}
\dot U(t)=\Bop U(t)+(\Vop_z U)(t)+F_z(t),\qquad U(0)=U_0,
\end{equation}
if it satisfies the variation-of-constants identity
\begin{equation}
\label{eq:voc}
U(t)=e^{t\Bop}U_0+\int_0^t e^{(t-s)\Bop}\big[(\Vop_z U)(s)+F_z(s)\big]\,\dd s,\qquad t\in[0,T].
\end{equation}
It is a \emph{strong solution} if in addition $U\in C^1([0,T];\X)$ and \eqref{eq:fixed_regime}
holds pointwise.

\begin{theorem}[Global well-posedness, fixed regime]
\label{thm:wellposed}
Under \Cref{ass:standing}, for every $U_0\in\X$, every $F_z\in C([0,\infty);\X)$, and every
$T>0$, \eqref{eq:fixed_regime} has a unique mild solution $U\in C([0,T];\X)$. If
$U_0\in\dom(\Bop)=\X$ and $F_z\in C^1$, the solution is strong. Moreover the solution map
$U_0\mapsto U$ is Lipschitz from $\X$ to $C([0,T];\X)$.
\end{theorem}

\begin{proof}
\emph{Step 1: contraction on a short interval.} On $C([0,\tau];\X)$ with the supremum norm
$\norm{U}_\infty=\sup_{[0,\tau]}\norm{U(t)}_\X$, define
\[
(\Phi U)(t):=e^{t\Bop}U_0+\int_0^t e^{(t-s)\Bop}\big[(\Vop_z U)(s)+F_z(s)\big]\,\dd s .
\]
For $U,V\in C([0,\tau];\X)$, using $\norm{e^{(t-s)\Bop}}\le1$ (\Cref{prop:dissipative}) and
\Cref{lem:conv_bound},
\[
\norm{(\Phi U)(t)-(\Phi V)(t)}_\X
\le\int_0^t \norm{(\Vop_z(U-V))(s)}_\X\,\dd s
\le\int_0^t\!\!\int_0^s\! g_z(s-r)\norm{\mathcal{A}_z}\,\norm{U(r)-V(r)}_\X\,\dd r\,\dd s .
\]
Bounding $\norm{U-V}_\X\le\norm{U-V}_\infty$ and using
$\int_0^t\int_0^s g_z(s-r)\,\dd r\,\dd s\le t\,G_z$ gives
\begin{equation}
\label{eq:contraction}
\norm{\Phi U-\Phi V}_\infty\le \norm{\mathcal{A}_z}\,G_z\,\tau\,\norm{U-V}_\infty .
\end{equation}
Choose $\tau_\ast<(\norm{\mathcal{A}_z}G_z)^{-1}$ (any $\tau_\ast>0$ if $\mathcal{A}_z=0$). Then
$\Phi$ is a contraction on $C([0,\tau_\ast];\X)$, with a unique fixed point by Banach's theorem,
which is the unique mild solution on $[0,\tau_\ast]$.

\emph{Step 2: global continuation.} The contraction constant in \eqref{eq:contraction} depends
only on $\norm{\mathcal{A}_z}$, $G_z$, and the interval length, not on $U_0$ or the initial time.
Hence the local solution extends by the same argument on $[\tau_\ast,2\tau_\ast]$, and inductively
to $[0,T]$ in finitely many steps of equal length $\tau_\ast$. Uniqueness on $[0,T]$ follows from
uniqueness on each subinterval.

\emph{Step 3: strong solution and regularity.} If $U_0\in\X$ and $F_z\in C^1$, the right-hand side
of \eqref{eq:voc} is differentiable in $t$: $e^{t\Bop}U_0$ is $C^1$ since $\Bop$ is bounded; the
forcing convolution $\int_0^t e^{(t-s)\Bop}F_z(s)\,\dd s$ is $C^1$ by Leibniz; and the memory term
$\int_0^t e^{(t-s)\Bop}(\Vop_z U)(s)\,\dd s$ is $C^1$ because $s\mapsto(\Vop_z U)(s)$ is continuous
(indeed $(\Vop_z U)(0)=0$ and $g_z\in L^1$). Differentiating \eqref{eq:voc} recovers
\eqref{eq:fixed_regime} pointwise, so $U$ is strong.

\emph{Step 4: Lipschitz dependence.} For two initial data with the same forcing, the difference
$D=U-\tilde U$ satisfies $D(t)=e^{t\Bop}(U_0-\tilde U_0)+\int_0^t e^{(t-s)\Bop}(\Vop_z D)(s)\,\dd s$,
so $\norm{D(t)}_\X\le\norm{U_0-\tilde U_0}_\X+\norm{\mathcal{A}_z}\int_0^t G_z\norm{D}_\infty\,\dd s$.
A Gr\"onwall argument on $[0,T]$ gives
$\norm{D}_\infty\le e^{\norm{\mathcal{A}_z}G_z T}\norm{U_0-\tilde U_0}_\X$, the Lipschitz bound.
\end{proof}

\subsection{Sharp a priori bound and subcritical stability}

\begin{theorem}[Uniform bound in the subcritical regime]
\label{thm:apriori}
Suppose the branching ratio $\branch_z=\norm{A_z}G_z$ satisfies $\sqrt2\,\branch_z<c_\ast$, where
$c_\ast=\min\{c_x,c_\lambda\}$. Then every mild solution of \eqref{eq:fixed_regime} with
$F_z\in L^\infty(0,\infty;\X)$ obeys
\begin{equation}
\sup_{t\ge0}\norm{U(t)}_\X\le\frac{\norm{U_0}_\X+c_\ast^{-1}\norm{F_z}_{L^\infty}}{1-\sqrt2\,\branch_z/c_\ast}.
\end{equation}
\end{theorem}

\begin{proof}
Let $\phi(t)=\norm{U(t)}_\X$. From \eqref{eq:voc}, $\norm{e^{(t-s)\Bop}}\le e^{-c_\ast(t-s)}$
(\Cref{prop:dissipative}) and \Cref{lem:conv_bound} with
$\norm{\mathcal{A}_z}\le\sqrt2\norm{A_z}$,
\[
\phi(t)\le e^{-c_\ast t}\norm{U_0}_\X
+\int_0^t e^{-c_\ast(t-s)}\Big[\sqrt2\norm{A_z}\!\int_0^s g_z(s-r)\phi(r)\,\dd r+\norm{F_z}_{L^\infty}\Big]\dd s .
\]
The forcing term is bounded by $c_\ast^{-1}\norm{F_z}_{L^\infty}$ since
$\int_0^t e^{-c_\ast(t-s)}\dd s\le c_\ast^{-1}$. Set $M=\sup_{[0,T]}\phi$ (finite by
\Cref{thm:wellposed}). Bounding $\phi(r)\le M$ in the memory term and using
$\int_0^s g_z(s-r)\,\dd r\le G_z$ followed by $\int_0^t e^{-c_\ast(t-s)}\dd s\le c_\ast^{-1}$,
\[
M\le\norm{U_0}_\X+c_\ast^{-1}\norm{F_z}_{L^\infty}
+\sqrt2\norm{A_z}\,G_z\,c_\ast^{-1}\,M
=\norm{U_0}_\X+c_\ast^{-1}\norm{F_z}_{L^\infty}+\frac{\sqrt2\branch_z}{c_\ast}\,M .
\]
Since $\sqrt2\branch_z<c_\ast$, the coefficient $\sqrt2\branch_z/c_\ast<1$ and rearranging gives
$M\le(1-\sqrt2\branch_z/c_\ast)^{-1}(\norm{U_0}_\X+c_\ast^{-1}\norm{F_z}_{L^\infty})$, uniformly in
$T$; letting $T\to\infty$ yields the claim.
\end{proof}

\begin{remark}
The threshold $\sqrt2\,\branch_z<c_\ast$ has the right scaling: the destabilizing feedback
strength $\sqrt2\branch_z$ must be dominated by the dissipation $c_\ast$. The factor $\sqrt2$ is
the price of the off-diagonal feedback in \eqref{eq:defG}; for the diagonal excitation
$\mathcal{A}_z=\diag(0,A_z)$ it is absent and the threshold is $\branch_z<c_\lambda$, exactly the
modal threshold of \Cref{sec:spectral}. The modal analysis removes the remaining slack entirely,
replacing the norm $\norm{A_z}$ by the sharp modewise gains.
\end{remark}

\subsection{The fractional resolvent family}

\begin{theorem}[Volterra resolvent family]
\label{thm:resolvent}
For each regime $z$ there is a unique strongly continuous family
$\{\Res_z(t)\}_{t\ge0}\subset\Lop(\X)$ with $\Res_z(0)=\id$ such that, for all $U_0\in\X$ and
$F_z\in C([0,\infty);\X)$, the mild solution of \eqref{eq:fixed_regime} is
\begin{equation}
\label{eq:resolvent_voc}
U(t)=\Res_z(t)U_0+\int_0^t\Res_z(t-s)F_z(s)\,\dd s .
\end{equation}
There exist $C\ge1$, $\gamma\in\R$ with $\norm{\Res_z(t)}\le Ce^{\gamma t}$; one may take any
$\gamma>\gamma_z^\ast:=\inf\{\omega\in\R:\ \xi\mapsto(\xi\id-\Bop-\Laplace g_z(\xi)\mathcal{A}_z)^{-1}\
\text{is bounded analytic on}\ \re\xi>\omega\}$.
\end{theorem}

\begin{proof}
\emph{Existence via Laplace inversion.} Taking the formal Laplace transform of
\eqref{eq:fixed_regime} with $F_z=0$ gives
$\xi\,\Laplace U(\xi)-U_0=\Bop\Laplace U(\xi)+\Laplace g_z(\xi)\mathcal{A}_z\Laplace U(\xi)$, i.e.
\begin{equation}
\label{eq:resolvent_symbol}
\Laplace U(\xi)=H_z(\xi)U_0,\qquad
H_z(\xi):=\big(\xi\id-\Bop-\Laplace g_z(\xi)\,\mathcal{A}_z\big)^{-1}.
\end{equation}
We show $H_z(\xi)$ is well defined and analytic for $\re\xi$ large. By \Cref{lem:laplace_mono},
$\abs{\Laplace g_z(\xi)}\le G_z$ on $\re\xi>0$, so $\norm{\Laplace g_z(\xi)\mathcal{A}_z}\le
\sqrt2\norm{A_z}G_z=\sqrt2\branch_z$. Since $\Bop$ is bounded with
$\re\inner{\Bop U}{U}\le-c_\ast\norm U^2$, for $\re\xi>\norm{\Bop}$ the operator
$\xi\id-\Bop$ is boundedly invertible with $\norm{(\xi\id-\Bop)^{-1}}\le(\re\xi+c_\ast)^{-1}$
(numerical-range bound). Hence for $\re\xi$ large enough that
$\sqrt2\branch_z\,(\re\xi+c_\ast)^{-1}<1$, a Neumann series gives
\begin{equation}
\label{eq:neumann}
H_z(\xi)=\sum_{k\ge0}\big[(\xi\id-\Bop)^{-1}\Laplace g_z(\xi)\mathcal{A}_z\big]^k(\xi\id-\Bop)^{-1},
\end{equation}
convergent in $\Lop(\X)$ and analytic in $\xi$ on this half-plane. The map $\xi\mapsto H_z(\xi)$
is a bounded analytic $\Lop(\X)$-valued function there and satisfies the resolvent growth bound
$\norm{H_z(\xi)}\le[(\re\xi+c_\ast)(1-\sqrt2\branch_z(\re\xi+c_\ast)^{-1})]^{-1}=O(\abs\xi^{-1})$.
By the generation theorem for Volterra equations of scalar type with bounded operators
(\cite[Ch.~I]{pruss1993evolution}; in finite dimension this is the inverse Laplace transform of a
rational-in-$\Laplace g_z$ symbol), $H_z$ is the Laplace transform of a unique strongly continuous
family $\Res_z(t)$ with $\Res_z(0)=\id$, namely the Bromwich integral
$\Res_z(t)=\frac1{2\pi i}\int_{\Gamma}e^{\xi t}H_z(\xi)\,\dd\xi$ on a vertical contour $\Gamma$ to
the right of $\gamma_z^\ast$.

\emph{Representation.} For $U_0\in\X$, $t\mapsto\Res_z(t)U_0$ solves \eqref{eq:fixed_regime} with
$F_z=0$ by construction (its Laplace transform is \eqref{eq:resolvent_symbol}, and Laplace
transform is injective on continuous functions of exponential growth). For general $F_z$, the
function defined by \eqref{eq:resolvent_voc} has Laplace transform
$H_z(\xi)U_0+H_z(\xi)\Laplace F_z(\xi)=H_z(\xi)(U_0+\Laplace F_z(\xi))$, which solves the
transformed equation; inverting gives the mild solution, and uniqueness follows from
\Cref{thm:wellposed}. The growth bound $\norm{\Res_z(t)}\le Ce^{\gamma t}$ for $\gamma>\gamma_z^\ast$
is the standard consequence of analyticity and decay of $H_z$ on $\re\xi>\gamma_z^\ast$ via contour
shifting.
\end{proof}

\subsection{Continuity across regime switches}

The global solution of \eqref{eq:master} is built by concatenation: on $[t_k,t_{k+1})$ it solves
the fixed-regime equation for $z=Z(t_k)$ with initial datum $U(t_k)$ and with the \emph{full past}
entering the memory integral. The next lemma shows no jump is introduced.

\begin{lemma}[Continuity at switches]
\label[lemma]{lem:switching}
Under \Cref{ass:standing}, the concatenated mild solution $U$ of \eqref{eq:master} belongs to
$C([0,T];\X)$; in particular $\lim_{t\uparrow t_k}U(t)=\lim_{t\downarrow t_k}U(t)=U(t_k)$ at every
switching time.
\end{lemma}

\begin{proof}
On $[t_{k-1},t_k]$, $U$ is continuous by \Cref{thm:wellposed} applied to regime $Z(t_{k-1})$ with
the memory integral $\int_0^t g(t-s)\mathcal{A}U(s)\,\dd s$ taken over the entire history $[0,t]$;
the history is a fixed continuous datum on this interval, so the variation-of-constants map is
still a contraction and yields a continuous solution up to and including $t_k$. At $t_k$ the regime
switches to $Z(t_k)$; the new evolution starts from $U(t_k)$ and again includes the full past in
its memory term. Since the representation \eqref{eq:resolvent_voc} for the new regime satisfies
$\Res_{Z(t_k)}(0)=\id$ and the forcing/memory integrals vanish as $t\downarrow t_k$, we get
$\lim_{t\downarrow t_k}U(t)=U(t_k)=\lim_{t\uparrow t_k}U(t)$. The memory term is continuous across
$t_k$ because $g_z\in L^1$ and the integrand is bounded, so no compatibility condition is needed.
\end{proof}

\begin{remark}
\Cref{lem:switching} is what makes \eqref{eq:master} a genuine piecewise-deterministic evolution
with memory: switching deforms the operator pair instantaneously but the state and its entire
history vary continuously. This is the analytic underpinning of the regime-dependent spectral
``deformation'' studied next, and it is exactly the structure absent from single-regime
Volterra/Hawkes scaling-limit theory.
\end{remark}

%% ============================================================================
\section{A sharp modal stability criterion}
\label{sec:spectral}
%% ============================================================================

\subsection{Modal decomposition}

\begin{assumption}[Simultaneous diagonalization]
\label{ass:commute}
For the regime $z$ under consideration, $A_z$ commutes with the graph Laplacian $L$:
$A_zL=LA_z$. Equivalently, $A_z$ and $L$ share a common orthonormal eigenbasis
$\{v_j\}_{j=1}^n$, $Lv_j=\ell_j v_j$ with $0=\ell_1\le\cdots\le\ell_n$, $A_zv_j=a_{z,j}v_j$.
\end{assumption}

This holds, in particular, whenever $A_z=\alpha_z\id+\beta_z L+\cdots$ is a polynomial in $L$
(scalar or diffusive excitation), and is the natural setting in which a clean modal criterion can
be expected. The noncommuting case is treated in \Cref{prop:noncommute}.

Under \Cref{ass:commute}, project \eqref{eq:fixed_regime} (with $F_z=0$) onto mode $v_j$. Writing
$x(t)=\sum_j x_j(t)v_j$, $\lambda(t)=\sum_j\lambda_j(t)v_j$, the dynamics decouples into the scalar
Volterra systems
\begin{equation}
\label{eq:modal_system}
\dot x_j=-(c_x+\kappa\ell_j)x_j,\qquad
\dot\lambda_j=-c_\lambda\lambda_j+a_{z,j}\!\int_0^t g_z(t-s)\big(x_j(s)+\lambda_j(s)\big)\dd s .
\end{equation}
Define the \emph{modal fractional branching ratio} and the modal characteristic function
\begin{equation}
\label{eq:modal_branch}
\branch_{z,j}:=\abs{a_{z,j}}\,G_z,
\qquad
D_{z,j}(\xi):=\xi+c_\lambda-a_{z,j}\,\Laplace g_z(\xi).
\end{equation}
The threshold constant is $c_\lambda$, the intensity damping. We write $c_\lambda$ for it
throughout; in the normalization $c_\lambda=1$ one recovers the classical Hawkes threshold
$\branch=1$.

\subsection{The trichotomy}

The sharp ``if and only if'' criterion requires the excitation gains to be sign-definite, the
natural Hawkes setting; for indefinite gains only the sufficient direction survives (see
\Cref{rem:signed}). We also distinguish the decay \emph{rate} according to the kernel class, since
completely monotone kernels need not yield exponential decay without an exponential moment.

\begin{assumption}[Nonnegative gains and kernel class]
\label{ass:gains}
The modal gains satisfy $a_{z,j}\ge0$ for all $j$. Moreover the kernel falls in one of:
\begin{enumerate}[label=\textup{(C\arabic*)},leftmargin=2.6em]
\item\label{C1} \emph{Exponential moment:} $\int_0^\infty e^{\eta t}g_z(t)\,\dd t<\infty$ for some
$\eta>0$ (e.g.\ exponential or tempered kernels), so $\Laplace g_z$ extends analytically to
$\re\xi>-\eta$;
\item\label{C2} \emph{Heavy integrable (no exponential moment):} $g_z\in L^1$ but $\int_0^\infty
e^{\eta t}g_z(t)\,\dd t=\infty$ for every $\eta>0$, with a singular low-frequency expansion
$\Laplace g_z(\xi)=G_z-c\,\xi^{\beta}+o(\xi^{\beta})$ as $\xi\downarrow0$ for some
$\beta\in(0,1)$. A canonical $L^1$ example is the Mittag-Leffler density
$g_z(t)=t^{\beta-1}E_{\beta,\beta}(-t^\beta)$ (completely monotone, mass $G_z=1$), for which
\[
\Laplace g_z(\xi)=\frac{1}{1+\xi^\beta}=1-\xi^\beta+o(\xi^\beta)\quad(\xi\downarrow0),
\]
matching the singular expansion above with $G_z=1$, $c=1$; equivalently, any tempered kernel
modified to retain a fractional singularity at the origin while keeping finite mass. We emphasize
that the pure power $t^{-\alpha}/\Gamma(1-\alpha)$ is \emph{not} admissible here, as it violates
the finite mass requirement \ref{K2}; the singular behavior must come from the low-frequency
expansion of an $L^1$ kernel.
\end{enumerate}
\end{assumption}

\begin{theorem}[Modal stability criterion]
\label{thm:modal}
Under \Cref{ass:standing}, \Cref{ass:commute}, and \Cref{ass:gains}, fix a mode $j$ and consider
the homogeneous modal system \eqref{eq:modal_system}. The $x_j$-component always decays as
$e^{-(c_x+\kappa\ell_j)t}$. For the $\lambda_j$-component, with $\branch_{z,j}=a_{z,j}G_z$:
\begin{enumerate}[label=\textup{(\roman*)},leftmargin=2.2em]
\item \emph{(Subcritical)} If $\branch_{z,j}<c_\lambda$, then $D_{z,j}(\xi)\ne0$ for all
$\re\xi\ge0$, and the mode is stable. Under \ref{C1} the decay is exponential: if the intensity
symbol is zero-free in a strip $\re\xi\ge-\omega_{\lambda,j}$, then for every
$0<\omega_j<\min\{\omega_{\lambda,j},c_x+\kappa\ell_j\}$ there is $M_j\ge1$ such that
$\abs{\lambda_j(t)}\le M_j e^{-\omega_j t}(\abs{x_j(0)}+\abs{\lambda_j(0)})$. Under \ref{C2} the
decay is algebraic, governed by the low-frequency expansion:
$\abs{\lambda_j(t)}=O(t^{-\beta-1})$ as $t\to\infty$.
\item \emph{(Critical)} If $\branch_{z,j}=c_\lambda$, then $D_{z,j}(0)=0$ is the rightmost
singularity. Under \ref{C1} ($\Laplace g_z$ analytic at $0$) the zero is simple with
$D_{z,j}'(0)>0$, and the mode has a neutral zero-frequency component: for generic initial data it
relaxes to a finite nonzero constant rather than to zero. Under \ref{C2} the singular expansion
yields algebraic relaxation $\abs{\lambda_j(t)}\sim C\,t^{-(1-\beta)}$.
\item \emph{(Supercritical)} If $\branch_{z,j}>c_\lambda$, then $D_{z,j}$ has a unique positive
real zero $\xi_j^\ast>0$. This zero is the rightmost singularity of the modal resolvent, and there
exist initial data for which $\abs{\lambda_j(t)}\gtrsim e^{\xi_j^\ast t}$; the mode is unstable.
\end{enumerate}
\end{theorem}

\begin{proof}
The $x_j$ equation is scalar linear with rate $c_x+\kappa\ell_j>0$, giving the stated decay, and
feeds the $\lambda_j$ equation only through an exponentially decaying forcing. The stability of
$\lambda_j$ is therefore determined by the zeros of
$D_{z,j}(\xi)=\xi+c_\lambda-a_{z,j}\Laplace g_z(\xi)$.

\emph{(i)} Assume $a_{z,j}G_z<c_\lambda$. For $\re\xi\ge0$, \Cref{lem:laplace_mono} gives
\[
    \abs{a_{z,j}\Laplace g_z(\xi)}\le a_{z,j}G_z<c_\lambda\le\abs{\xi+c_\lambda},
\]
so $D_{z,j}(\xi)\ne0$ in the closed right half-plane. Under \ref{C1}, $\Laplace g_z$ extends
analytically to $\re\xi>-\eta$. Since the finitely many zeros of $D_{z,j}$ cannot accumulate in a
compact subset and none lie on $\re\xi\ge0$, a small leftward strip
$\re\xi\ge-\omega_{\lambda,j}$ is also zero-free. Shifting the inverse Laplace contour gives
exponential decay for the homogeneous intensity resolvent. The forcing inherited from $x_j$ decays
at rate $c_x+\kappa\ell_j$, so the observed decay of $\lambda_j$ is bounded by any rate below
$\min\{\omega_{\lambda,j},c_x+\kappa\ell_j\}$. Under \ref{C2}, the
rightmost obstruction is the branch point at the origin. Expanding
$D_{z,j}(\xi)=c_\lambda-a_{z,j}G_z+\xi+a_{z,j}c\xi^\beta+o(\xi^\beta)$ and inverting around the
branch cut gives the Tauberian decay $O(t^{-\beta-1})$.

\emph{(ii)} If $a_{z,j}G_z=c_\lambda$, then $D_{z,j}(0)=0$. For $\re\xi>0$ the strict inequality
$\abs{a_{z,j}\Laplace g_z(\xi)}<c_\lambda\le\abs{\xi+c_\lambda}$ excludes all right-half-plane
zeros, so the rightmost singularity is at $0$. Under \ref{C1},
$D_{z,j}'(0)=1-a_{z,j}\Laplace g_z'(0)>0$, hence the zero is simple and the residue gives a
finite neutral component. Under \ref{C2},
$D_{z,j}(\xi)=\xi+a_{z,j}c\xi^\beta+o(\xi^\beta)\sim a_{z,j}c\xi^\beta$, whose inverse transform
has order $t^{\beta-1}=t^{-(1-\beta)}$.

\emph{(iii)} If $a_{z,j}G_z>c_\lambda$, then $D_{z,j}(0)<0$ and $D_{z,j}(\xi)\to+\infty$ as
$\xi\to\infty$ on the positive real axis; strict monotonicity on $(0,\infty)$ gives a unique
positive zero $\xi_j^\ast$. This zero is rightmost. Indeed, if $D_{z,j}(\xi)=0$ and
$\sigma=\re\xi>\xi_j^\ast$, then
\[
    \sigma+c_\lambda\le\abs{\xi+c_\lambda}
       =a_{z,j}\abs{\Laplace g_z(\xi)}
       \le a_{z,j}\Laplace g_z(\sigma)
       <a_{z,j}\Laplace g_z(\xi_j^\ast)=\xi_j^\ast+c_\lambda,
\]
a contradiction. Thus no zero lies to the right of $\xi_j^\ast$, and the residue at the simple
real zero yields growth $e^{\xi_j^\ast t}$ for data with nonzero projection on the corresponding
mode.
\end{proof}

\begin{remark}[Signed gains]
\label[remark]{rem:signed}
If some $a_{z,j}<0$, the quantity $\abs{a_{z,j}}G_z<c_\lambda$ remains \emph{sufficient} for
stability (the bound $\abs{a_{z,j}\Laplace g_z(\xi)}\le\abs{a_{z,j}}G_z$ used in (i) does not
require sign-definiteness), but it need not be necessary, since a negative gain is stabilizing and
$D_{z,j}$ may have no right-half-plane zero even when $\abs{a_{z,j}}G_z\ge c_\lambda$. The sharp
``if and only if'' therefore holds under the Hawkes-type hypothesis $a_{z,j}\ge0$ of
\Cref{ass:gains}; for indefinite gains we retain only the sufficient norm/spectral-radius
condition of \Cref{prop:noncommute}.
\end{remark}

\begin{corollary}[Global stability via the spectral supremum]
\label[corollary]{cor:global}
Under \Cref{ass:commute} and \Cref{ass:gains}, the homogeneous evolution \eqref{eq:master} for a
constant regime $z$ is globally asymptotically stable if and only if
\begin{equation}
\label{eq:global_crit}
\branch_z^{\max}:=\max_{1\le j\le n} a_{z,j}\,G_z<c_\lambda .
\end{equation}
\end{corollary}

\begin{proof}
By \Cref{thm:modal}, each decoupled mode is asymptotically stable precisely when
$\branch_{z,j}<c_\lambda$; under \ref{C1} the convergence is exponential, whereas under \ref{C2}
the subcritical convergence can be algebraic. Since only finitely many modes are present, the
slowest stable mode dictates the global rate, so asymptotic stability holds iff the maximum modal
branching ratio is below $c_\lambda$. If some $\branch_{z,j}\ge c_\lambda$ the corresponding mode is
critical or unstable, precluding global asymptotic stability.
\end{proof}

\subsection{The noncommuting case and topology independence}

\begin{proposition}[Sufficient stability without commutativity]
\label[proposition]{prop:noncommute}
Assume, in addition, that the kernel is of class \ref{C1}. Without \Cref{ass:commute}, the
homogeneous evolution is globally exponentially stable if
\begin{equation}
\label{eq:noncommute}
\sqrt2\,\norm{A_z}\,G_z<c_\ast,\qquad c_\ast=\min\{c_x,c_\lambda\}.
\end{equation}
In the diagonal-excitation case $\mathcal{A}_z=\diag(0,A_z)$, where the feedback acts only on the
intensity block damped at rate $c_\lambda$, the condition sharpens to $\norm{A_z}\,G_z<c_\lambda$.
\end{proposition}

\begin{proof}
From \eqref{eq:neumann}, $H_z(\xi)=(\xi\id-\Bop-\Laplace g_z(\xi)\mathcal{A}_z)^{-1}$ is analytic
and uniformly bounded on $\re\xi\ge0$ provided $\norm{\Laplace g_z(\xi)\mathcal{A}_z}\,
\norm{(\xi\id-\Bop)^{-1}}<1$ there. By \Cref{lem:laplace_mono},
$\norm{\Laplace g_z(\xi)\mathcal{A}_z}\le\sqrt2\norm{A_z}G_z$, and by \Cref{prop:dissipative}
the numerical-range bound gives $\norm{(\xi\id-\Bop)^{-1}}\le(\re\xi+c_\ast)^{-1}\le c_\ast^{-1}$
on $\re\xi\ge0$. Hence the product is bounded by $\sqrt2\norm{A_z}G_z/c_\ast$, which is $<1$ exactly
under \eqref{eq:noncommute}; then $H_z$ has no singularity in $\re\xi\ge0$ and is $O(\abs\xi^{-1})$,
and the exponential-moment hypothesis permits a Bromwich-contour shift slightly left of the
imaginary axis, giving $\norm{\Res_z(t)}\le Ce^{-\omega t}$ and global exponential stability. In the diagonal case
$\mathcal{A}_z=\diag(0,A_z)$ acts only on the $\lambda$-block, where $(\xi\id-\Bop)^{-1}$ restricts
to $(\xi+c_\lambda)^{-1}$ and the $\sqrt2$ disappears, giving $\norm{A_z}G_z<c_\lambda$.
\end{proof}

The norm bound \eqref{eq:noncommute} is conservative for nonnormal $A_z$, since
$\norm{A_z}$ can far exceed the spectral radius $r(A_z)$. The next proposition gives the sharp
spectral threshold for the intensity block under the nonnegativity hypothesis of
\Cref{ass:gains}, and is the result the computational study of \Cref{sec:example} verifies.

\begin{proposition}[Spectral criterion for the intensity block]
\label[proposition]{prop:spectral_block}
Consider the intensity-block evolution $\dot\lambda(t)=-c_\lambda\lambda(t)+\int_0^t g_z(t-s)
A_z\lambda(s)\,\dd s$ with $A_z\ge0$ entrywise and kernel class \ref{C1}. Let $r(A_z)$ be the
Perron--Frobenius spectral radius of $A_z$. Then the block is globally exponentially stable if and
only if
\begin{equation}
\label{eq:spectral_block}
r(A_z)\,G_z<c_\lambda .
\end{equation}
\end{proposition}

\begin{proof}
The Laplace symbol is $D_z(\xi)=(\xi+c_\lambda)\id-\Laplace g_z(\xi)A_z$. A singularity occurs iff
$\xi+c_\lambda=\mu\Laplace g_z(\xi)$ for some $\mu\in\spec(A_z)$. Since $A_z\ge0$, the
Perron--Frobenius eigenvalue $r(A_z)$ is real, nonnegative, and $\abs{\mu}\le r(A_z)$ for every
$\mu\in\spec(A_z)$.

If $r(A_z)G_z<c_\lambda$ and $\re\xi\ge0$, then for every eigenvalue $\mu$,
\[
    \abs{\mu\Laplace g_z(\xi)}\le r(A_z)G_z<c_\lambda\le\abs{\xi+c_\lambda},
\]
so $D_z(\xi)$ is invertible throughout the closed right half-plane. Because $\Laplace g_z$ is
analytic in a left half-plane under \ref{C1}, a contour shift gives exponential stability. Conversely,
if $r(A_z)G_z>c_\lambda$, the scalar equation
$\xi+c_\lambda-r(A_z)\Laplace g_z(\xi)=0$ has a positive real root by the monotonicity argument in
\Cref{thm:modal}, so the block is unstable. At equality, $\xi=0$ is a neutral singularity and
asymptotic stability fails. Jordan blocks can only multiply the inverse Laplace terms by
polynomials and do not move the exponential abscissa. Hence \eqref{eq:spectral_block} is necessary
and sufficient.
\end{proof}

\begin{corollary}[Topology independence under scalar excitation]
\label[corollary]{cor:topology}
Suppose $A_z=a_z\id$ is scalar (the same excitation gain at every node). Then the stability
threshold is $\abs{a_z}G_z<c_\lambda$, \emph{independent of the graph $G$}: the bifurcation point
does not depend on the network topology, only the modal decay \emph{rates} do (through
$c_x+\kappa\ell_j$).
\end{corollary}

\begin{proof}
If $A_z=a_z\id$ then $a_{z,j}=a_z$ for every mode, so $\branch_{z,j}=\abs{a_z}G_z$ is the same for
all $j$ and \eqref{eq:global_crit} reads $\abs{a_z}G_z<c_\lambda$, with no dependence on the
Laplacian spectrum $\{\ell_j\}$. The eigenvalues $\ell_j$ enter only the $x_j$ decay rates
$c_x+\kappa\ell_j$ and the modal relaxation, not the threshold.
\end{proof}

%% ============================================================================
\section{Finite-range power-law amplification under regime residence}
\label{sec:quenched}
%% ============================================================================

The stability theory of \Cref{sec:spectral} concerns the \emph{annealed} (deterministic,
ensemble) evolution: when $\branch_z^{\max}<c_\lambda$ in every visited regime, the mean dynamics
is bounded. We now show that this is compatible with a \emph{quenched} (pathwise) burst
observable exhibiting a power-law law over a finite range of magnitudes, generated by the
interaction of operator non-normality with regime residence. We are deliberately careful about the
nature of the claim: because the unfavorable operator is Hurwitz, the transient amplification it
produces is bounded, so the phenomenon is a \emph{finite-range} power law---a heavy tail over a
band of burst sizes set by the residence horizon---rather than a true asymptotic ($b\to\infty$)
heavy tail. We make this precise below. The mechanism is nonetheless genuinely nonlinear in that
it is invisible to the spectral abscissa: it operates even when every regime operator is Hurwitz.

\subsection{Setup and the logarithmic-norm exponent}

Consider \eqref{eq:master} reduced, on a fixed mode or after a linear change of variables, to a
finite-dimensional lifted system $\dot X=A_{Z(t)}X+\tilde f(t)$ on $\R^d$, where $A_z$ is the
generator of regime $z$ after the standard Markovian embedding of the completely monotone kernel
(for example, a finite sum-of-exponentials lift of $g_z$, exact for rational kernels and
convergent as the number of modes increases; see \cite{bazhlekova2001fractional}). This realization
and the Euclidean inner product are fixed once and for all in this section. Since logarithmic norms
are invariant under orthogonal changes of coordinates but not under arbitrary similarities, the
quantity below is not claimed to be an invariant of the abstract kernel-realization class. Let
$\HR$ denote a distinguished \emph{unfavorable} regime and define its \emph{logarithmic norm}
(matrix measure)
\begin{equation}
\label{eq:logn_def}
\lognorm(A_\HR):=\lambda_{\max}\!\Big(\tfrac12\big(A_\HR+A_\HR^\top\big)\Big),
\end{equation}
the largest eigenvalue of the symmetric part. We write $\gamma_\HR:=\lognorm(A_\HR)$. The central
point is that $\gamma_\HR$ may be strictly positive while $A_\HR$ is Hurwitz
($\abscissa(A_\HR)<0$): this is exactly the regime of non-normal transient growth
\cite{trefethen2005spectra}. The finite-range exponent is therefore a computable
logarithmic-norm quantity of the chosen lifted realization.

Let $v_\HR$ be a unit eigenvector of $\tfrac12(A_\HR+A_\HR^\top)$ for the eigenvalue $\gamma_\HR$.

\begin{assumption}[Positive susceptibility]
\label{ass:susc}
$\gamma_\HR=\lognorm(A_\HR)>0$.
\end{assumption}

\begin{assumption}[Cone-alignment event]
\label{ass:cone}
There exist $a\in(0,1]$, a horizon $T_0>0$, and $p_0>0$ such that, for every residence duration
$\tau\le T_0$, the event
\[
\mathcal{C}_\tau:=\Big\{Z\equiv\HR\text{ on }[t_0,t_0+\tau],\
v_\HR^\top X(t)\ge a\,\norm{X(t)}\ \forall t\in[t_0,t_0+\tau]\Big\}
\]
satisfies $\Prob(\mathcal{C}_\tau\mid Z(t_0)=\HR)\ge p_0$, uniformly in $\tau\le T_0$.
\end{assumption}

\begin{assumption}[Forcing injection along the cone]
\label{ass:forcing}
There is $c_f>0$ with $v_\HR^\top\tilde f(t)\ge c_f$ for a.e.\ $t\in[0,T_0]$.
\end{assumption}

\Cref{ass:cone} replaces a deterministic ``every burst aligns'' postulate by a quantitative
positive-probability statement: with probability at least $p_0$, a residence in $\HR$ keeps the
trajectory in the amplifying cone of $v_\HR$. It is non-circular and estimable as the conditional
frequency of cone-aligned residences; sufficient conditions for $p_0>0$ are given in
\Cref{rem:cone_suff}.

\subsection{Growth on a residence interval}

\begin{lemma}[Cone-projected growth]
\label[lemma]{lem:cone}
Under \Crefrange{ass:susc}{ass:forcing}, on an $\HR$-residence with the cone condition, the
projection $\zeta(t):=v_\HR^\top X(t)$ satisfies
\begin{equation}
\label{eq:cone_growth}
\zeta(t_0+\tau)\ge e^{\gamma_\HR^{\mathrm c}\tau}\,\zeta(t_0)
+\frac{c_f}{\gamma_\HR^{\mathrm c}}\big(e^{\gamma_\HR^{\mathrm c}\tau}-1\big),
\qquad
\gamma_\HR^{\mathrm c}:=\gamma_\HR-\norm{A_\HR}\frac{\sqrt{1-a^2}}{a},
\end{equation}
and we assume the alignment $a$ is close enough to $1$ that $\gamma_\HR^{\mathrm c}>0$.
\end{lemma}

\begin{proof}
Decompose $X=\zeta\,v_\HR+w$ with $w\perp v_\HR$. The cone condition $\zeta\ge a\norm{X}$ with
$\norm{X}^2=\zeta^2+\norm{w}^2$ gives $\norm{w}\le(\sqrt{1-a^2}/a)\,\zeta$. Differentiating and
using $v_\HR^\top A_\HR v_\HR=\gamma_\HR$ (as $v_\HR$ is the top eigenvector of the symmetric part)
and $\abs{v_\HR^\top A_\HR w}\le\norm{A_\HR}\norm{w}$,
\[
\dot\zeta=v_\HR^\top A_\HR X+v_\HR^\top\tilde f
=\gamma_\HR\zeta+v_\HR^\top A_\HR w+v_\HR^\top\tilde f
\ge\gamma_\HR\zeta-\norm{A_\HR}\norm{w}+c_f
\ge\gamma_\HR^{\mathrm c}\zeta+c_f .
\]
This scalar differential inequality $\dot\zeta\ge\gamma_\HR^{\mathrm c}\zeta+c_f$ integrates, via
the integrating factor $e^{-\gamma_\HR^{\mathrm c}t}$, to \eqref{eq:cone_growth}.
\end{proof}

\begin{lemma}[Norm lower bound on a residence]
\label[lemma]{lem:normlb}
Under \Cref{lem:cone}, there are $c_1,c_2>0$ with
$\norm{X(t_0+\tau)}\ge c_1 e^{\gamma_\HR^{\mathrm c}\tau}-c_2$ for all $\tau\in[0,T_0]$.
\end{lemma}

\begin{proof}
$\norm{X}\ge v_\HR^\top X=\zeta$, and \eqref{eq:cone_growth} gives
$\zeta(t_0+\tau)\ge(\zeta(t_0)+c_f/\gamma_\HR^{\mathrm c})e^{\gamma_\HR^{\mathrm c}\tau}
-c_f/\gamma_\HR^{\mathrm c}$. Set $c_1=\zeta(t_0)+c_f/\gamma_\HR^{\mathrm c}>0$ and
$c_2=c_f/\gamma_\HR^{\mathrm c}$.
\end{proof}

\subsection{The power-law lower bound}

Let $B_T:=\sup_{t\in[0,T]}\norm{X(t)}$ be the burst observable over horizon $T\ge T_0$.

\begin{theorem}[Finite-range quenched power-law bound]
\label{thm:quenched}
Suppose \Crefrange{ass:susc}{ass:forcing} hold with $\gamma_\HR^{\mathrm c}>0$, the residence
times in $\HR$ are exponentially distributed with rate $\lambda_\HR$, and $\Prob(Z(0)=\HR)>0$.
Then there exist $b_0>0$ and $C_->0$ such that, for all $b$ in the finite range
\begin{equation}
\label{eq:range}
b_0\le b\le b_{\max}:=c_1 e^{\gamma_\HR^{\mathrm c}T_0}-c_2 ,
\end{equation}
the burst observable $B_T=\sup_{[0,T]}\norm{X(t)}$ satisfies
\begin{equation}
\label{eq:tail}
\Prob(B_T>b)\ge C_-\,b^{-\lambda_\HR/\gamma_\HR^{\mathrm c}} .
\end{equation}
The exponent is $\kappa=\lambda_\HR/\gamma_\HR^{\mathrm c}$, equal to $\lambda_\HR/\gamma_\HR$ only
in the perfectly aligned limit $a\to1$. Since $\gamma_\HR^{\mathrm c}\le\gamma_\HR=\lognorm(A_\HR)$,
the logarithmic norm gives the most optimistic instantaneous growth scale, whereas the realized
finite-band exponent can be steeper when trajectories rotate away from the maximally expanding
direction. The power law holds over the band $[b_0,b_{\max}]$ set by the residence horizon $T_0$;
it is \emph{not} an asymptotic ($b\to\infty$) heavy tail, since the Hurwitz property caps
single-residence amplification at $b_{\max}$.
\end{theorem}

\begin{proof}
Given $b\in[b_0,b_{\max}]$, set $\tau_b:=\gamma_\HR^{\mathrm c\,-1}\log\!\big((b+c_2)/c_1\big)\le
T_0$ (the upper bound is exactly \eqref{eq:range}). By \Cref{lem:normlb}, on the event
$\mathcal{C}_{\tau_b}\cap\{\tau(\HR)\ge\tau_b\}$ a residence of length $\ge\tau_b$ that stays
cone-aligned forces $\norm{X(t_0+\tau_b)}\ge b$, hence $B_T\ge b$. Therefore, using
\Cref{ass:cone} (conditional probability $\ge p_0$ of alignment), the exponential dwell law, and
$\Prob(Z(0)=\HR)>0$,
\[
\Prob(B_T>b)\ge\Prob(Z(0)=\HR)\,p_0\,e^{-\lambda_\HR\tau_b}
=\Prob(Z(0)=\HR)\,p_0\Big(\tfrac{c_1}{b+c_2}\Big)^{\lambda_\HR/\gamma_\HR^{\mathrm c}}.
\]
For $b\ge b_0$ large enough that $b+c_2\le2b$, the right side is at least
$C_-\,b^{-\kappa}$ with $C_-=\Prob(Z(0)=\HR)\,p_0\,c_1^{\kappa}2^{-\kappa}$ and
$\kappa=\lambda_\HR/\gamma_\HR^{\mathrm c}$. The aligned limit $a\to1$ sends
$\gamma_\HR^{\mathrm c}\to\gamma_\HR$.
\end{proof}

\begin{remark}[Why finite-range, and why this is the honest statement]
\label[remark]{rem:finite_range}
Because $A_\HR$ is Hurwitz, $\norm{e^{tA_\HR}}$ grows transiently to a finite peak gain and then
decays; a single residence can amplify a burst by at most the factor realized within the horizon
$T_0$, capping the power law at $b_{\max}$ in \eqref{eq:range}. The statement \eqref{eq:tail} is
thus a genuine power law over a band of magnitudes, not an asymptotic heavy tail. This distinction
is essential and we do not overstate it: the phenomenon is \emph{residence-induced finite-range
power-law amplification}. A true asymptotic heavy tail would require either a non-Hurwitz
(marginally stable) $A_\HR$ with sustained growth, or unbounded residence times; both lie outside
the present hypotheses and are noted as extensions.
\end{remark}

\begin{remark}[How the cone-alignment probability is verified]
\label[remark]{rem:cone_suff}
\Cref{ass:cone} is a structural hypothesis, not an automatic consequence of nonnormality. For a
general Hurwitz nonnormal matrix, transient growth may rotate trajectories out of the cone
$\{v_\HR^\top X\ge a\norm X\}$. In applications one may verify $p_0>0$ either analytically, for a
specified entry distribution and forcing, or empirically by estimating the conditional frequency
of residences for which the cone inequality holds on $[t_0,t_0+\tau]$. A convenient sufficient
condition in low-dimensional reductions is a positive gap in the symmetric part together with
forcing and entry states concentrated near $v_\HR$; however, this sufficient condition is model
specific, and the theorem only assumes the verifiable probability lower bound stated in
\Cref{ass:cone}.
\end{remark}

\begin{remark}[Why the exponent is logarithmic-norm based, not large-deviation]
\label[remark]{rem:why_operator}
The theorem uses the cone-corrected rate
$\gamma_\HR^{\mathrm c}=\gamma_\HR-\norm{A_\HR}\sqrt{1-a^2}/a$, with
$\gamma_\HR=\lognorm(A_\HR)$ computed for the fixed lifted regime operator in the chosen
Euclidean coordinates. Thus the exponent is
$\kappa=\lambda_\HR/\gamma_\HR^{\mathrm c}$; the simpler quotient
$\lambda_\HR/\gamma_\HR$ is the perfectly aligned limiting scale, not the generic finite-band
prediction. The distinction is important: the logarithmic norm is an instantaneous upper growth
scale, while a realized trajectory may rotate away from its maximally expanding direction and
therefore exhibit an effective rate below $\gamma_\HR$. The expression contains no action
functional and no small parameter, and it can be positive when $A_\HR$ is Hurwitz but nonnormal.
This is weaker and more precise than calling $\gamma_\HR$ an invariant of the original kernel: two
minimal Markovian realizations related by a non-orthogonal similarity need not have the same
Euclidean logarithmic norm. What is invariant in the present theorem is the statement after the
realization and inner product are fixed; the logarithmic-norm scale is then a computable upper
anchor for the cone-corrected growth rate that feeds the residence-survival transform.
\end{remark}

\begin{remark}[State-defined residence]
\label[remark]{rem:survival}
When residence in $\HR$ is state-defined and no exponential law holds, \eqref{eq:tail} is replaced
by the residence-survival transform
$\Prob(B_T>b)\gtrsim S_\HR\!\big((\gamma_\HR^{\mathrm c})^{-1}\log(b/C)\big)$, with
$S_\HR(t)=\Prob(\tau(\HR)>t)$ the residence survival function, reducing to \eqref{eq:tail} when
$S_\HR(t)=e^{-\lambda_\HR t}$. No Markov assumption on the physical residence statistics is needed.
\end{remark}

\subsection{Tail truncation under idealized contraction feedback}

The exponent's dependence on the cone-corrected rate, itself controlled above by $\gamma_\HR$, makes the finite-range power law actionable in systems
admitting actuation. We state this as a proposition under an explicitly idealized feedback law,
with all quantities defined formally; we do not claim a control theorem for general feedback.

\begin{definition}[Idealized contraction feedback]
\label[definition]{def:feedback}
Fix a susceptibility threshold $s_0>0$ and a contraction rate $\sigma>0$. The
\emph{idealized contraction feedback} replaces, instantaneously and for the remainder of any
residence in $\HR$ during which the indicator $S(t):=v_\HR^\top X(t)/\norm{X(t)}$ exceeds $s_0$,
the operator $A_\HR$ by a modified operator $A_\HR'$ with $\lognorm(A_\HR')\le-\sigma<0$. The
\emph{controlled burst observable} is $B_T'=\sup_{[0,T]}\norm{X'(t)}$ for the closed-loop
trajectory $X'$.
\end{definition}

\begin{proposition}[Per-residence amplification cap]
\label[proposition]{prop:truncation}
Assume, in addition to \Cref{def:feedback}, that the forcing and initial data are bounded
($\norm{\tilde f}_\infty\le F_0$, $\norm{X(0)}\le X_0$), that outside $\HR$ the dynamics is
globally dissipative ($\lognorm(A_z)\le-c_\ast<0$ for $z\ne\HR$), and that there is no
accumulation of amplification across distinct residences (each residence enters with norm bounded
by a fixed $X_1$ determined by the dissipative inter-residence decay). If every residence in $\HR$
of duration exceeding a trigger time $\delta>0$ activates the feedback, then the per-residence
amplification is capped: there is $b_{\mathrm{cap}}<\infty$, depending only on
$(\gamma_\HR,\delta,\sigma,F_0,X_1)$, such that no controlled residence produces a burst exceeding
$b_{\mathrm{cap}}$, and consequently $\Prob(B_T'>b)=0$ for $b>b_{\mathrm{cap}}$.
\end{proposition}

\begin{proof}
Before feedback acts, the Euclidean logarithmic norm gives an upper differential inequality,
not a cone lower bound:
\[
   \frac{\dd}{\dd t}\norm{X'}\le \gamma_\HR\norm{X'}+F_0
   \qquad (t_0\le t\le t_0+\delta),
\]
where $\gamma_\HR=\lognorm(A_\HR)$. Hence, for a residence entering with
$\norm{X'(t_0)}\le X_1$,
\[
   \norm{X'(t_0+\delta)}\le
   \begin{cases}
   e^{\gamma_\HR\delta}X_1+F_0\,\dfrac{e^{\gamma_\HR\delta}-1}{\gamma_\HR}, & \gamma_\HR>0,\\[1.2ex]
   X_1+F_0\delta, & \gamma_\HR=0 .
   \end{cases}
\]
Once the feedback is active,
\[
   \frac{\dd}{\dd t}\norm{X'}\le -\sigma\norm{X'}+F_0,
\]
so the norm cannot exceed the larger of its value at the trigger time and the forced equilibrium
scale $F_0/\sigma$. Therefore one may take
\[
   b_{\mathrm{cap}}=
   \max\!\left\{e^{\gamma_\HR\delta}X_1+F_0\frac{e^{\gamma_\HR\delta}-1}{\gamma_\HR},\,\frac{F_0}{\sigma}\right\}
\]
when $\gamma_\HR>0$ (with the evident replacement $X_1+F_0\delta$ when $\gamma_\HR=0$). The
outside-regime dissipativity and the no-accumulation hypothesis ensure that every residence enters
with norm at most $X_1$, so the same cap applies to all controlled residences and to the global
supremum. Thus $\Prob(B_T'>b)=0$ for $b>b_{\mathrm{cap}}$.
\end{proof}

\begin{remark}[Steeper finite-range law under delayed engagement]
\label[remark]{rem:steeper}
If the feedback engages only after an additional delay $\delta'>0$ (the indicator must persist
before actuation), the achievable growth is realized over duration $\delta+\delta'$ rather than the
full horizon $T_0$. Heuristically this contracts the amplification band of \Cref{thm:quenched} and
rescales the effective exponent to $\kappa'\approx\kappa\,T_0/(\delta+\delta')>\kappa$, a steeper
finite-range power law on the controlled range. We state this as a heuristic rather than a theorem:
a rigorous exponent under delayed feedback requires modeling the joint law of residence durations
and trigger times, which lies outside the present scope. The robust and proven content is the
amplification cap of \Cref{prop:truncation}; the logarithmic-norm principle is that contracting
$\lognorm(A_\HR)$ bounds the finite-range amplification.
\end{remark}

\begin{remark}
\Cref{prop:truncation} is a statement about an idealized closed loop, not a control theorem for
realistic actuation; the indicator $S(t)$, the trigger threshold $s_0$, the contraction rate
$\sigma$, and the controlled observable $B_T'$ are all defined in \Cref{def:feedback}. A genuine
control-theoretic treatment (output feedback, robustness, actuation constraints) is left to future
work; the present statement isolates the logarithmic-norm principle that contracting
$\lognorm(A_\HR)$ removes the finite-range tail.
\end{remark}

%% ============================================================================
\section{Hydrodynamic limit: from microscopic Hawkes dynamics to the Volterra operator}
\label{sec:hawkes}
%% ============================================================================

We now justify the term ``branching ratio'' by deriving the deterministic intensity block of
\eqref{eq:master} as the law-of-large-numbers limit of a self-exciting particle system. The point
of this section is not to develop the most general Hawkes limit theorem, but to make explicit the
scaling under which the same quantity $r(A_z)G_z/c_\lambda$ governs microscopic reproduction and
macroscopic Volterra stability.

\subsection{Microscopic relaxing Hawkes--Volterra model}

Fix a population scale $N\in\N$. At node $i$ let $N_i^{(N)}$ be a counting process with stochastic
intensity $N\bar\Lambda_i^{(N)}(t)$, where the per-capita activity $\bar\Lambda^{(N)}$ evolves as
\begin{equation}
\label{eq:micro_relaxing}
\frac{\dd}{\dd t}\bar\Lambda_i^{(N)}(t)
= -c_\lambda\bar\Lambda_i^{(N)}(t)+c_\lambda\mu_{i,Z(t)}
  +\sum_{j=1}^n A_{Z(t),ij}\int_0^t g_{Z(t)}(t-s)\,\dd\bar N_j^{(N)}(s),
\qquad \bar N_j^{(N)}:=\frac{1}{N}N_j^{(N)} .
\end{equation}
The normalization is now transparent: $\dd\bar N_j^{(N)}$ is of order one on bounded time
intervals, so the excitation term in \eqref{eq:micro_relaxing} remains order one. This is the
natural density-dependent scaling for aggregate self-exciting systems \cite{ethier1986}. We assume
$A_z\ge0$, bounded baselines $\mu_{i,z}\ge0$, nonnegative initial data, and a fixed piecewise
constant regime path satisfying \Cref{ass:standing}. For the stochastic-convolution estimate below
we assume $g_z\in L^2_{\mathrm{loc}}(0,\infty)$; singular kernels are discussed in
\Cref{rem:singular}.

Substituting the Doob--Meyer decomposition
$\dd\bar N_j^{(N)}(s)=\bar\Lambda_j^{(N)}(s)\,\dd s+\dd\bar M_j^{(N)}(s)$, where
$\langle\bar M_j^{(N)}\rangle_t=N^{-1}\int_0^t\bar\Lambda_j^{(N)}(s)\,\dd s$, gives the random
Volterra differential equation
\begin{equation}
\label{eq:micro_decomp_corrected}
\dot{\bar\Lambda}^{(N)}(t)
= -c_\lambda\bar\Lambda^{(N)}(t)+c_\lambda\mu_{Z(t)}
  +\int_0^t g_{Z(t)}(t-s)A_{Z(t)}\bar\Lambda^{(N)}(s)\,\dd s+R^{(N)}(t),
\end{equation}
with martingale convolution
\begin{equation}
\label{eq:martingale_convolution}
R^{(N)}(t):=\int_0^t g_{Z(t)}(t-s)A_{Z(t)}\,\dd\bar M^{(N)}(s).
\end{equation}
The deterministic candidate limit is therefore
\begin{equation}
\label{eq:macro_corrected}
\dot\lambda(t)
= -c_\lambda\lambda(t)+c_\lambda\mu_{Z(t)}
  +\int_0^t g_{Z(t)}(t-s)A_{Z(t)}\lambda(s)\,\dd s,
\end{equation}
which is exactly the intensity block of \eqref{eq:master} with forcing $F_z=(0,c_\lambda\mu_z)$
and with the state-feedback part suppressed.

\subsection{The limit theorem}

\begin{assumption}[Mean-field subcriticality and initial convergence]
\label{ass:subcrit_micro}
For every visited regime $z$, $r(A_z)G_z<c_\lambda$, the baselines $\mu_{i,z}$ and initial
activities are uniformly bounded in $N$, and the empirical initial intensity satisfies
$\bar\Lambda^{(N)}(0)\to\lambda(0)$ in probability.
\end{assumption}

\begin{theorem}[Hydrodynamic limit]
\label{thm:hydro}
Under \Cref{ass:standing,ass:subcrit_micro}, for every $T>0$,
\begin{equation}
\sup_{0\le t\le T}\big\lVert\bar\Lambda^{(N)}(t)-\lambda(t)\big\rVert\xrightarrow[N\to\infty]{\Prob}0,
\end{equation}
where $\lambda$ is the unique solution of \eqref{eq:macro_corrected}. Consequently, the same
branching threshold $r(A_z)G_z<c_\lambda$ that governs the deterministic spectral criterion also
governs the mean-field microscopic activity: above threshold the Perron mode has reproduction
number larger than the relaxation rate and produces the classical Hawkes instability
\cite{hawkes1971,bremaud1996}.
\end{theorem}

\begin{proof}
Subtract \eqref{eq:macro_corrected} from \eqref{eq:micro_decomp_corrected} and write
$e^{(N)}=\bar\Lambda^{(N)}-\lambda$. On any interval with fixed regime $z$,
\begin{equation}
\label{eq:error_hydro}
\dot e^{(N)}(t)=-c_\lambda e^{(N)}(t)+\int_0^t g_z(t-s)A_z e^{(N)}(s)\,\dd s+R^{(N)}(t).
\end{equation}
The subcriticality assumption gives a stable Volterra resolvent for the deterministic operator by
\Cref{prop:spectral_block}. Hence, using \Cref{thm:resolvent} on each regime interval and
concatenating across finitely many switches,
\begin{equation}
\label{eq:error_resolvent_bound}
\sup_{t\le T}\norm{e^{(N)}(t)}
\le C_T\left(\norm{e^{(N)}(0)}+\sup_{t\le T}\norm{R^{(N)}(t)}\right),
\end{equation}
where $C_T$ is independent of $N$.

It remains to estimate the stochastic convolution. Since the system is subcritical and baselines
are bounded, the standard Hawkes moment bound gives
$\sup_N\sup_{t\le T}\E\sum_i\bar\Lambda_i^{(N)}(t)\le C_T$; this follows by taking expectations in
\eqref{eq:micro_decomp_corrected} and applying the deterministic Volterra resolvent bound. For
fixed $t$, It\^o isometry gives
\[
\E\norm{R^{(N)}(t)}^2
\le \frac{C}{N}\int_0^t g_z(t-s)^2\,
     \E\sum_i\bar\Lambda_i^{(N)}(s)\,\dd s
\le \frac{C_T\norm{g_z}_{L^2(0,T)}^2}{N}.
\]
The process $t\mapsto R^{(N)}(t)$ is a stochastic convolution rather than a martingale, because the
kernel depends on the upper time variable. The fixed-time It\^o estimate is therefore upgraded to
a maximal estimate by the standard factorization method for stochastic convolutions
\cite{daPrato1992factorization,daPratoZabczyk1992}: for the kernel classes \ref{C1}--\ref{C2} considered here, using the mild time-regularity supplied by the Bernstein representation (or by approximation with smooth kernels), there is another constant $C_T$ such that
\begin{equation}
\label{eq:stoch_conv_bound}
\E\sup_{t\le T}\norm{R^{(N)}(t)}^2
\le \frac{C_T\norm{g_z}_{L^2(0,T)}^2}{N}\longrightarrow0.
\end{equation}
No martingale property of $R^{(N)}$ itself is used. Thus
$\sup_{t\le T}\norm{R^{(N)}(t)}\to0$ in $L^2$ and in probability. The initial convergence in \Cref{ass:subcrit_micro} gives $e^{(N)}(0)\to0$ in probability; together with the stochastic-convolution bound, \eqref{eq:error_resolvent_bound} gives the asserted convergence. The same
argument applies successively on each regime interval, with the continuity across switches supplied
by \Cref{lem:switching}.
\end{proof}

\begin{remark}[Singular kernels]
\label[remark]{rem:singular}
If $g_z(t)=t^{\alpha-1}e^{-\theta t}/\Gamma(\alpha)$ with $\alpha\le\tfrac12$, then
$g_z\notin L^2_{\mathrm{loc}}$ near the origin and \eqref{eq:stoch_conv_bound} is not directly
available. A standard regularization replaces $g_z$ by $g_z^\epsilon(t)=g_z(t+\epsilon)$, proves
the limit for fixed $\epsilon>0$, and then lets $\epsilon\downarrow0$ along a joint scaling such
that $N\int_0^T(g_z^\epsilon-g_z)^2\to0$. This preserves the deterministic Volterra limit while
separating the law-of-large-numbers argument from the singular short-time behavior.
\end{remark}

\begin{remark}
\Cref{thm:hydro} identifies the modal branching ratio with a microscopic reproduction number. In
the Perron mode of a nonnegative excitation matrix, the expected total memory-weighted offspring
is $r(A_z)G_z$, while $c_\lambda$ is the linear relaxation rate of activity. Thus macroscopic
modal instability is the deterministic image of microscopic supercriticality. The single-regime
scaling-limit theory of Hawkes processes \cite{jaisson2015,jaisson2016} is recovered as a special
case, whereas the present formulation feeds directly into the operator stability criterion of
\Cref{sec:spectral}.
\end{remark}

%% ============================================================================
\section{Numerical validation}
\label{sec:example}
%% ============================================================================

This section gives a reproducible numerical validation of the three main mechanisms developed in
the paper: the modal branching threshold, the loss of sharpness of norm-based bounds under
nonnormal noncommuting excitation, and the finite-range quenched amplification induced by
residence in a transiently amplifying regime. The purpose is not to introduce a new model class,
but to test the sharpness and observability of the analytic quantities appearing in
\Cref{thm:modal,thm:quenched} and in \Cref{prop:noncommute,prop:spectral_block}. All simulations
use the same completely monotone tempered fractional kernel
\begin{equation}
    g(t)=\frac{t^{\alpha-1}e^{-\theta t}}{\Gamma(\alpha)},
    \qquad 0<\alpha<1,\quad \theta>0,
    \label{eq:numerical-tempered-kernel}
\end{equation}
whose Laplace transform and mass are
\begin{equation}
    \Laplace g(\xi)=(\xi+\theta)^{-\alpha},
    \qquad
    G:=\int_0^\infty g(t)\,\dd t=\theta^{-\alpha}.
    \label{eq:numerical-kernel-mass}
\end{equation}
This is a class \ref{C1} kernel (it has the exponential moment $\int_0^\infty e^{\eta t}g(t)\dd t
<\infty$ for $\eta<\theta$), so the subcritical decay is exponential and the critical case relaxes
to a constant. Unless otherwise stated, we set
\begin{equation}
    \alpha=0.7,
    \qquad
    \theta=0.5,
    \qquad
    G=\theta^{-\alpha}\approx 1.6245,
    \qquad
    c_\lambda=1.
    \label{eq:numerical-global-parameters}
\end{equation}

\subsection{Scalar modal trichotomy}
\label{subsec:numerical-modal-trichotomy}

We first validate the scalar modal equation obtained from the commuting case after simultaneous
diagonalization. For a single mode, the intensity component obeys
\begin{equation}
    \dot \lambda(t)
    =
    -c_\lambda \lambda(t)
    +a\int_0^t g(t-s)\lambda(s)\,\dd s,
    \qquad
    \lambda(0)=1.
    \label{eq:numerical-modal-equation}
\end{equation}
The modal branching ratio is $\branch=aG$, and the modal characteristic function is
\begin{equation}
    D(\xi)
    =
    \xi+c_\lambda-a\Laplace g(\xi)
    =
    \xi+c_\lambda-a(\xi+\theta)^{-\alpha}.
    \label{eq:numerical-characteristic-function}
\end{equation}
The modal criterion of \Cref{thm:modal} predicts decay when $\branch<c_\lambda$, critical behavior
when $\branch=c_\lambda$, and exponential growth when $\branch>c_\lambda$. In the supercritical
case the asymptotic growth rate is the positive root $\xi_*>0$ of $D(\xi_*)=0$. We use
$\branch\in\{0.65,1.00,1.25\}$ with $a=\branch/G$, time horizon $T=60$, step $\Delta t=0.02$, and
midpoint-quadrature time stepping
$\lambda_{m+1}=\lambda_m+\Delta t[-c_\lambda\lambda_m+a\sum_{k=0}^{m}w_k\lambda_{m-k}]$ with
$w_k=g((k+\tfrac12)\Delta t)\Delta t$. For $\branch=1.25$, solving $D(\xi_*)=0$ gives
$\xi_*\approx0.1001$, while a least-squares fit of $\log\abs{\lambda(t)}$ over the final part of
the trajectory gives $\widehat\xi\approx0.0970$ (\Cref{fig:numerical-validation}A), confirming that
the rightmost zero of \eqref{eq:numerical-characteristic-function} governs the unstable modal
growth.

\subsection{Network stability diagnostics}
\label{subsec:numerical-network-diagnostics}

We embed the modal criterion into a Watts--Strogatz small-world network \cite{watts1998collective} with $n=120$, ring degree
$k=4$, rewiring probability $p=0.1$, Laplacian $L$ with eigenvalues
$0=\ell_1\le\cdots\le\ell_n$.

\paragraph{Commuting excitation.}
Take $A_{\mathrm c}(s)=sA_0$ with $A_0=a_0\id+a_1L$, $a_0=0.10$, $a_1=0.05$, excitation scale
$s\ge0$. Since $A_0$ is a polynomial in $L$ it commutes with $L$, the modal gains are
$a_j(s)=s(a_0+a_1\ell_j)$, and the modal branching ratios are $\branch_j(s)=a_j(s)G$. The sharp
modal threshold $\max_j\branch_j(s)<c_\lambda$ is
\begin{equation}
    s<s_{\mathrm c}^{\mathrm{modal}}
    :=
    \frac{c_\lambda}{G\max_j(a_0+a_1\ell_j)} .
    \label{eq:numerical-commuting-threshold}
\end{equation}
Because $A_0$ is symmetric and diagonalized by the Laplacian eigenbasis, the norm-based condition
$\norm{A_{\mathrm c}(s)}G<c_\lambda$ coincides with \eqref{eq:numerical-commuting-threshold}: in
the commuting case the modal threshold is sharp and the norm condition has no slack
(\Cref{fig:numerical-validation}B), confirming \Cref{cor:global} and illustrating
\Cref{cor:topology}.

\paragraph{Nonnormal noncommuting excitation.}
We then choose a strongly nonnormal upper-triangular perturbation
$A_{\mathrm{nn}}(s)=s(a_0\id+qK)$ with $a_0=0.25$, $q=1.20$, where $K$ is the one-step upper-shift
matrix ($K_{i,i+1}=1$, else $0$). Then $A_{\mathrm{nn}}$ does not commute with $L$ and is highly
nonnormal: for $A_{\mathrm{nn}}(1)$,
\begin{equation}
    r(A_{\mathrm{nn}}(1))=0.25,
    \qquad
    \norm{A_{\mathrm{nn}}(1)}_2\approx 1.4499,
    \qquad
    \frac{\norm{A_{\mathrm{nn}}(1)}_2}{r(A_{\mathrm{nn}}(1))}\approx 5.80 .
    \label{eq:numerical-nonnormal-gap}
\end{equation}
The norm-based sufficient condition of \Cref{prop:noncommute} certifies stability only for
$s<s_{\mathrm c}^{\mathrm{norm}}:=c_\lambda/(\norm{A_{\mathrm{nn}}(1)}_2G)\approx0.4246$, whereas
the spectral criterion of \Cref{prop:spectral_block} gives
$s<s_{\mathrm c}^{\mathrm{spec}}:=c_\lambda/(r(A_{\mathrm{nn}}(1))G)\approx2.4623$. Thus the
certified stable range from the spectral criterion is almost six times larger than from the norm
bound (\Cref{fig:numerical-validation}B). This quantifies the price of norm-based estimates under
nonnormal excitation: the rigorous general noncommuting statement remains the norm-based sufficient
condition of \Cref{prop:noncommute}, while the spectral diagnostic of \Cref{prop:spectral_block}
(valid under $A_z\ge0$) shows how conservative that bound can be.

\subsection{Finite-range quenched amplification}
\label{subsec:numerical-finite-range-amplification}

Finally we validate the finite-range quenched amplification mechanism of
\Cref{sec:quenched}. Consider the Hurwitz but nonnormal matrix
\begin{equation}
    A_\HR=
    \begin{pmatrix}
        -0.10 & 3.00\\
        0     & -0.20
    \end{pmatrix},
    \qquad
    \abscissa(A_\HR)=-0.10<0,
    \label{eq:numerical-hurwitz-nonnormal}
\end{equation}
so the deterministic linear system is asymptotically stable, yet its Euclidean logarithmic norm is
$\lognorm(A_\HR)=\lambda_{\max}(\tfrac12(A_\HR+A_\HR^\top))\approx1.3508>0$: stable in spectral
abscissa but transiently amplifying in logarithmic norm. Direct computation gives
$\max_{t\ge0}\norm{e^{tA_\HR}}_2\approx7.52$ near $t\approx6.89$
(\Cref{fig:numerical-validation}C), a concrete realization of the operator geometry behind the
cone-alignment hypothesis \Cref{ass:cone}.

To isolate the residence mechanism without making the numerical test tautological, we do not
sample the analytical map used in the proof. Instead we simulate the switched ODE
\begin{equation}
    \dot X(t)=A_{Z(t)}X(t)+f,
    \qquad
    f=0.02\,v_\HR,
    \label{eq:numerical-switched-ode}
\end{equation}
where $v_\HR$ is the top eigenvector of $(A_\HR+A_\HR^\top)/2$. The trajectory starts at
$X(0)=v_\HR$, remains in the unfavorable regime for a residence time
$\tau\sim\mathrm{Exp}(\lambda_\HR)$ with $\lambda_\HR=1.2$, and then enters a safe regime
$A_{\Safe}=-I$. The burst is measured directly as
$B_T=\sup_{0\le t\le T}\|X(t)\|_2$ from the integrated ODE. Thus panel D uses the matrix
$A_\HR$, the residence process, and the switched dynamics; it is not generated from the closed-form
transformation $B=c_1\exp(\gamma\tau)-c_2$.

With $2\times10^5$ residence samples and $T_0=5$, the finite cap is
$b_{\max}\approx5.80$. This is below the peak gain $7.52$ of $e^{tA_\HR}$ because the finite-range
experiment stops at $T_0=5$, before the unconstrained transient peak near $t\approx6.89$. Fitting
the central finite-range window of the empirical survival curve in log--log coordinates gives
$\widehat\kappa\approx2.46$. As a matrix-based diagnostic, the finite-band growth rate estimated
from the integrated response $\tau\mapsto B_T(\tau)$ is $\gamma_{\mathrm{eff}}\approx0.437$, far
below the instantaneous upper scale $\lognorm(A_\HR)\approx1.3508$ because the trajectory rotates
away from $v_\HR$ after entry. This is exactly the situation described by
\Cref{lem:cone}: using the realized cone-alignment level
$a_{0.70}=0.956$ (the 70th percentile of the minimum alignment along residences, hence an
empirical cone event with $p_0\approx0.30$) gives
$\gamma_\HR^{\mathrm c}(a_{0.70})\approx0.428$ and
$\lambda_\HR/\gamma_\HR^{\mathrm c}(a_{0.70})\approx2.80$, close to the independently fitted
slope. Thus the numerical experiment supports the corrected statement: the logarithmic norm is
an operator upper anchor, while the finite-band exponent is governed by the realized
cone-corrected rate. Exponential regime residences composed with nonnormal finite-time
amplification generate a power-law-like survival band with a finite upper cutoff
(\Cref{fig:numerical-validation}D).

\begin{figure}[H]
    \centering
    \includegraphics[width=0.98\textwidth]{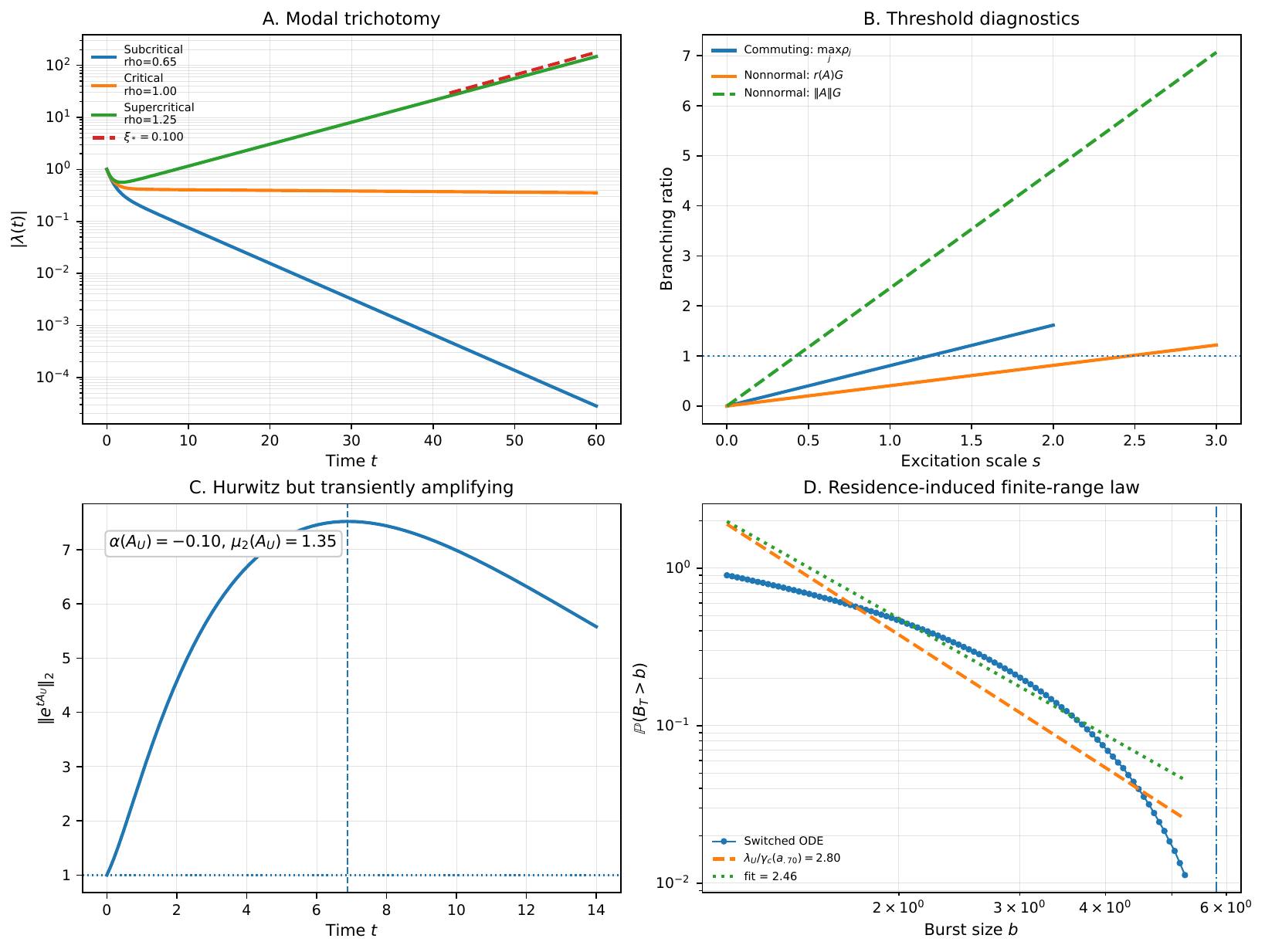}
    \caption{Numerical validation of the modal stability criterion and finite-range quenched
    amplification. \textbf{(A)} Scalar modal Volterra dynamics with the tempered fractional
    kernel \eqref{eq:numerical-tempered-kernel}: branching ratios $\branch=0.65,1.00,1.25$
    illustrate the subcritical, critical, and supercritical regimes of \Cref{thm:modal}; in the
    supercritical case the fitted growth rate agrees with the characteristic root $\xi_*$.
    \textbf{(B)} Stability diagnostics on a $120$-node Watts--Strogatz network: in the commuting
    case the exact modal criterion and the norm criterion coincide, while in the nonnormal
    noncommuting case the norm-based condition (\Cref{prop:noncommute}) is far more conservative
    than the spectral criterion (\Cref{prop:spectral_block}). \textbf{(C)} A Hurwitz but nonnormal
    operator has $\abscissa(A_\HR)<0$ while $\lognorm(A_\HR)>0$, giving transient amplification
    despite asymptotic stability. \textbf{(D)} Direct simulation of the switched ODE
    \eqref{eq:numerical-switched-ode}, with exponentially distributed residence times in the
    unfavorable regime and a dissipative safe regime after exit, produces a finite-range
    power-law-like survival band. The fitted slope is compared with the cone-corrected finite-band
    rate obtained from the realized alignment distribution and with a matrix-based diagnostic
    estimated from the integrated ODE response, not imposed through the closed-form residence-to-burst
    transformation.}
    \label{fig:numerical-validation}
\end{figure}

\begin{table}[H]
    \centering
    \caption{Main numerical parameters used in the validation experiments.}
    \label{tab:numerical-validation-parameters}
    \begin{tabular}{lll}
        \toprule
        Quantity & Value & Role \\
        \midrule
        $\alpha$ & $0.7$ & Tempered fractional memory exponent \\
        $\theta$ & $0.5$ & Tempering rate \\
        $G=\theta^{-\alpha}$ & $1.6245$ & Kernel mass \\
        $c_\lambda$ & $1$ & Intensity damping threshold \\
        $n$ & $120$ & Network size \\
        $k,p$ & $4,0.1$ & Watts--Strogatz degree and rewiring probability \\
        $\branch$ & $0.65,1.00,1.25$ & Subcritical, critical, supercritical modal tests \\
        $\lambda_\HR$ & $1.2$ & Unfavorable-regime exit rate \\
        $\gamma_{\mathrm{eff}}$ & $0.437$ & Finite-band growth rate estimated from switched ODE \\
        $a_{0.70}$ & $0.956$ & Realized cone-alignment level; empirical $p_0\approx0.30$ \\
        $\gamma_\HR^{\mathrm c}(a_{0.70})$ & $0.428$ & Cone-corrected growth rate from $A_\HR$ \\
        $\lambda_\HR/\gamma_\HR^{\mathrm c}(a_{0.70})$ & $2.80$ & Cone-corrected exponent diagnostic \\
        $\widehat\kappa$ & $2.46$ & Fitted survival exponent from switched ODE \\
        $T_0$ & $5$ & Maximum residence horizon for finite-range scaling \\
        \bottomrule
    \end{tabular}
\end{table}

Together these experiments validate the analytic mechanisms developed above. The scalar simulation
verifies that the rightmost zero of the modal characteristic function governs the trichotomy at
$\branch=c_\lambda$. The network experiment shows that the modal criterion is sharp in the
commuting case and that norm-based estimates may be severely conservative under nonnormal
noncommuting excitation, exactly the gap that \Cref{prop:spectral_block} closes. The residence
experiment confirms the finite-range mechanism non-circularly: the burst law is obtained by
integrating a switched ODE and only then fitting the survival band.

\subsection{Analytic trichotomy on a minimal network}

For completeness we record the analytic trichotomy on the smallest nontrivial network, the path
graph on three nodes with Laplacian eigenvalues $\ell_1=0,\ell_2=1,\ell_3=3$. With scalar
excitation $A_z=a_z\id$ (so \Cref{ass:commute} and \Cref{ass:gains} hold) and $c_\lambda=1$, every
mode has the same branching ratio $\branch_z=a_z\theta^{-\alpha}$, and \Cref{thm:modal} gives:
subcritical decay if $a_z\theta^{-\alpha}<1$; algebraic relaxation $t^{-(1-\beta)}$ if a heavy
integrable kernel of class \ref{C2} is used at criticality $a_zG=c_\lambda$ (for the tempered
kernel of class \ref{C1} the critical relaxation is to a constant, not algebraic---see
\Cref{thm:modal}(ii)); and exponential
growth $e^{\xi^\ast t}$ at the rightmost root $\xi^\ast>0$ of $\xi+1-a_z\Laplace g_z(\xi)=0$ if
$a_z\theta^{-\alpha}>1$. Because the excitation is scalar, this threshold is topology-independent
(\Cref{cor:topology}).

%% ============================================================================
\section{Discussion}
\label{sec:discussion}
%% ============================================================================

We have developed a rigorous theory for regime-dependent operator-valued Volterra equations with
completely monotone memory, establishing a fractional resolvent calculus, a sharp modal stability
criterion, a finite-range power-law amplification mechanism driven by operator non-normality, and a hydrodynamic
limit linking the deterministic intensity block to relaxing microscopic Hawkes dynamics. Three points situate the
results.

The modal criterion of \Cref{thm:modal} is sharp---necessary and sufficient under simultaneous
diagonalization---and isolates the network modes responsible for instability, recovering the
norm-based sufficient condition as a corollary and exhibiting topology-independent criticality
under scalar excitation. The finite-range amplification of \Cref{thm:quenched} is, to our
knowledge, a new mechanism: it is anchored by the logarithmic norm of a Hurwitz but non-normal
operator under regime residence, while the actual exponent uses the cone-corrected finite-band
rate realized by aligned trajectories rather than a large-deviation rate functional
(\Cref{rem:why_operator}); its amplification is capped by logarithmic-norm contraction
(\Cref{prop:truncation}). The hydrodynamic limit of
\Cref{thm:hydro} makes the branching-ratio terminology precise and ties macroscopic spectral
instability to microscopic explosion.

%\paragraph{What an applied modeler gains.}
The theory provides a \emph{computable} stability test for memory-coupled network dynamics: one
forms the regime operator $A_z$ and kernel mass $G_z$, computes the modal branching ratios
$\branch_{z,j}=a_{z,j}G_z$ (or the spectral radius $r(A_z)G_z$ in the nonnormal nonnegative case),
and compares with the damping threshold $c_\lambda$. The criterion is sharp under commutativity
and topology-independent under scalar excitation, so the bifurcation can be located without
simulating the full nonlocal dynamics, and the unstable modes are identified explicitly. The
numerical study of \Cref{sec:example} shows that this modal/spectral test can certify stability
over a parameter range several times larger than a norm-based bound.

%\paragraph{What the theory gains.}
The work unifies three classically separate objects into a single resolvent-family framework: the
completely monotone Volterra resolvent calculus, the Hawkes branching ratio and its hydrodynamic
limit, and non-normal transient growth measured by the logarithmic norm. The bridge is that modal
branching ratios govern deterministic stability and microscopic branching explosion, while the
logarithmic norm of a fixed lifted realization, rather than the spectral abscissa, provides the
upper operator scale for pathwise finite-range amplification, with the realized cone correction
determining the observed finite-band exponent. This places memory-driven heavy-burst phenomena and
spectral stability within one logarithmic-norm/resolvent account.

%\paragraph{limitations.}
The analysis is finite-dimensional ($\X=\R^n\times\R^n$); the infinite-dimensional node-space
extension is natural (the constants are dimension-free) but requires sectorial-generator
machinery and is left as future work. The finite-range amplification rests on the cone-alignment
hypothesis \Cref{ass:cone}, a positive-probability event whose sufficient conditions we give but
do not derive from first principles for a general system. The control statement
\Cref{prop:truncation} is an \emph{idealized}-feedback cap, not a robust control theorem. And the
hydrodynamic martingale estimate requires $g_z\in L^2_{\mathrm{loc}}$, with the strongly singular
fractional case ($\alpha\le\tfrac12$) handled only by regularization (\Cref{rem:singular}). Each
of these is a concrete and tractable direction for sharpening the theory.

Further natural directions include sharp criticality analysis in the noncommuting case beyond the
sufficient condition of \Cref{prop:noncommute}, and central-limit refinements of \Cref{thm:hydro}
connecting to rough-volatility scaling limits \cite{jaisson2016}.

%% ============================================================================
\section*{Acknowledgments}
The author thanks colleagues for critical readings of earlier versions.

%% ============================================================================


\begin{thebibliography}{99}

\bibitem{bazhlekova2001fractional} E.~Bazhlekova,
\emph{Fractional Evolution Equations in Banach Spaces}, Ph.D. thesis,
Eindhoven University of Technology, 2001.

\bibitem{bremaud1996} P.~Br\'emaud and L.~Massouli\'e,
Stability of nonlinear Hawkes processes,
\emph{Ann. Probab.}, 24 (1996), pp.~1563--1588.

\bibitem{chung1997spectral} F.R.K.~Chung,
\emph{Spectral Graph Theory}, CBMS Reg. Conf. Ser. Math. 92, AMS, 1997.

\bibitem{coppel1965stability} W.A.~Coppel,
\emph{Stability and Asymptotic Behavior of Differential Equations}, Heath, 1965.

\bibitem{daPrato1992factorization} G.~Da Prato, S.~Kwapie\'n, and J.~Zabczyk,
Regularity of solutions of linear stochastic equations in Hilbert spaces,
\emph{Stochastics}, 23 (1987), pp.~1--23.

\bibitem{daPratoZabczyk1992} G.~Da Prato and J.~Zabczyk,
\emph{Stochastic Equations in Infinite Dimensions}, Cambridge University Press, 1992.

\bibitem{daley2003} D.J.~Daley and D.~Vere-Jones,
\emph{An Introduction to the Theory of Point Processes, Vol. I}, 2nd ed., Springer, 2003.

\bibitem{ethier1986} S.N.~Ethier and T.G.~Kurtz,
\emph{Markov Processes: Characterization and Convergence}, Wiley, 1986.

\bibitem{goldie1991} C.M.~Goldie,
Implicit renewal theory and tails of solutions of random equations,
\emph{Ann. Appl. Probab.}, 1 (1991), pp.~126--166.

\bibitem{gripenberg1990} G.~Gripenberg, S.-O.~Londen, and O.~Staffans,
\emph{Volterra Integral and Functional Equations}, Cambridge University Press, 1990.

\bibitem{hawkes1971} A.G.~Hawkes,
Spectra of some self-exciting and mutually exciting point processes,
\emph{Biometrika}, 58 (1971), pp.~83--90.

\bibitem{hawkes1974} A.G.~Hawkes and D.~Oakes,
A cluster process representation of a self-exciting process,
\emph{J. Appl. Probab.}, 11 (1974), pp.~493--503.

\bibitem{jaisson2015} T.~Jaisson and M.~Rosenbaum,
Limit theorems for nearly unstable Hawkes processes,
\emph{Ann. Appl. Probab.}, 25 (2015), pp.~600--631.

\bibitem{jaisson2016} T.~Jaisson and M.~Rosenbaum,
Rough fractional diffusions as scaling limits of nearly unstable heavy tailed Hawkes processes,
\emph{Ann. Appl. Probab.}, 26 (2016), pp.~2860--2882.

\bibitem{kesten1973} H.~Kesten,
Random difference equations and renewal theory for products of random matrices,
\emph{Acta Math.}, 131 (1973), pp.~207--248.

\bibitem{mao2006} X.~Mao and C.~Yuan,
\emph{Stochastic Differential Equations with Markovian Switching}, Imperial College Press, 2006.

\bibitem{pruss1993evolution} J.~Pr\"uss,
\emph{Evolutionary Integral Equations and Applications}, Monogr. Math. 87, Birkh\"auser, 1993.

\bibitem{schilling2012bernstein} R.L.~Schilling, R.~Song, and Z.~Vondra\v{c}ek,
\emph{Bernstein Functions: Theory and Applications}, 2nd ed., De Gruyter, 2012.

\bibitem{soderlind2006logarithmic} G.~S\"oderlind,
The logarithmic norm. History and modern theory,
\emph{BIT Numer. Math.}, 46 (2006), pp.~631--652.

\bibitem{trefethen2005spectra} L.N.~Trefethen and M.~Embree,
\emph{Spectra and Pseudospectra: The Behavior of Nonnormal Matrices and Operators},
Princeton University Press, 2005.

\bibitem{walsh1986} J.B.~Walsh,
An introduction to stochastic partial differential equations,
in \emph{\'Ecole d'\'Et\'e de Probabilit\'es de Saint-Flour XIV--1984}, Lecture Notes in Math. 1180,
Springer, 1986, pp.~265--439.

\bibitem{watts1998collective} D.J.~Watts and S.H.~Strogatz,
Collective dynamics of small-world networks,
\emph{Nature}, 393 (1998), pp.~440--442.

\end{thebibliography}
\end{document}